\documentclass{article}
\usepackage{amsmath}
\usepackage{amssymb}
\usepackage{amsthm}

\newtheorem{theorem}{Theorem}[section]
\newtheorem{lemma}[theorem]{Lemma}
\newtheorem{remark}{Remark}[section]

\newtheorem{innercustomgeneric}{\customgenericname}
\providecommand{\customgenericname}{}
\newcommand{\newcustomtheorem}[2]{%
  \newenvironment{#1}[1]
  {%
   \renewcommand\customgenericname{\bf #2}%
   \renewcommand\theinnercustomgeneric{ ##1}%
   \innercustomgeneric
  }
  {\endinnercustomgeneric}
}

\newcustomtheorem{customthm}{Theorem}
\newcustomtheorem{customlemma}{{\bf Lemma}}

\usepackage{graphicx}
\usepackage{comment}
\usepackage{xcolor}
\usepackage{cite}
\usepackage{subcaption}
\usepackage{authblk}
\usepackage[subrefformat=parens, labelfont=up]{subcaption}

\title{The effect of host population heterogeneity on epidemic outbreaks}
\author[1,2,*]{M. C. J. Bootsma}
\affil[1]{Department of Mathematics, Faculty of Science, Utrecht University, the Netherlands}
\affil[2]{Julius Center for Health Sciences and Primary Care, University Medical Center Utrecht, Utrecht University, the Netherlands}
\author[3,4]{K. M. D. Chan}
\affil[3]{Korteweg-de Vries Institute, University of Amsterdam, the Netherlands}
\affil[4]{Transtrend BV, Rotterdam, The Netherlands}
\author[1]{O. Diekmann}
\author[5]{H. Inaba}
\affil[5]{Faculty of Education, Tokyo Gakugei University, Koganei-shi, Tokyo, Japan}
\affil[*]{Corresponding author: M.C.J.Bootsma@uu.nl}
\begin{document}

\maketitle
%    Abstract is required.
\begin{abstract}
In the first part of this paper, we review old and new results about the influence of host population heterogeneity on (various characteristics of) epidemic outbreaks. In the second part we highlight a modelling issue that so far has received little attention: how do contact patterns, and hence transmission opportunities, depend on the size and the composition of the host population? Without any claim on completeness, we offer a range of potential (quasi-mechanistic) submodels. The overall aim of the paper is to describe the state-of-the-art and to catalyse new work.
\end{abstract}

%=====================Main body of the paper ======
\section{Introduction}

As the title indicates, our aim is to investigate how heterogeneity influences various aspects of an epidemic outbreak in a demographically closed host population. In particular, we focus on the Basic Reproduction Number $R_0$, the Malthusian parameter $r$, the final size, the timing and the size of the peak in incidence.

The paper consists of two parts. In the first part we adopt a top down approach. We first introduce a rather general model. Since its formulation employs measures over the trait space, it covers discrete and continuous traits in one go. We present various results that, essentially, have been known for a long time, but that might not exactly be ‘well known’. Next we consider various simplifications and their underlying interpretation and motivation. These lead us to recent results, for instance on the HIT (herd immunity threshold; see Section \ref{section_sep2}), that were triggered by the outbreak of Covid-19.

An important, yet sometimes rather implicit, ingredient of epidemic models is a specification of the rate at which an individual, with a certain trait, makes contact with other individuals having a specified trait. With this in mind, Mossong et al. \cite{Mossong2008} studied age-structured contact patterns empirically, based on a population survey in various European countries. In the context of theoretical ‘what if’ studies, it is crucial to extrapolate such quantitative information. A concrete example is the final size of an outbreak of a new influenza type in Hong Kong \cite{Chan2016} and the question of what to expect in the future when, due to ageing, the age-distribution has changed considerably. Another example is motivated by Corona ticket measures: How does the contact process at a venue change when unvaccinated individuals are banned and the vaccination coverage is age-dependent? The aim of the second part of this paper (Sections \ref{section_IPS} and \ref{section_exampleIPS}) is to initiate theoretical work on contact patterns by presenting various motivated options of how contact intensities might depend on the size and composition of the population and how this affects the final size of the outbreak.

For a general introduction to structured epidemic models, see Chapters 7, 8 and 9 of \cite{Diekmann2013} and see \cite{Inaba2017, Li2020}.

%\part{}

\section{Model formulation}

The host population we consider is assumed to be static with respect to demography, so neither birth nor death of host individuals is taken into account. Host individuals are characterized by a trait, denoted by symbols like $x$ and $\xi$ and sometimes called ‘type’, rather than trait, in particular in situations when there exist only finitely many types. The trait is static, i.e., does not change during the life of the host individual, but see Appendix \ref{appendix_dynamic}. The trait $x$ takes values in a set $\Omega$, which is a measurable space, i.e., $\Omega$ is equipped with a $\sigma$-algebra. For concreteness, we let $\Omega$ be a subset of $\mathbb R^n$ with the Borel $\sigma$-algebra. The distribution of $x$ within the host population is described by the (given/known) measure $\Phi$, in order to unify the treatment of a ‘continuum’ setting, with a trait distribution described by a density, and the ‘discrete’ setting with finitely many types. The host population size is denoted by $N$. So $\Phi(\Omega)=1$ and the number of individuals with trait belonging to the measurable subset $\omega$ of $\Omega$ equals $N \Phi(\omega)$. Apart from $N$ and $\Phi$ we need just one modelling ingredient: 
\begin{equation}
    \begin{aligned}
         A(\tau,x,\xi) =& {\rm ~the ~expected ~contribution ~to ~the ~force ~of ~infection ~on ~an ~individual ~with ~trait ~{\it x}}\cr 
         &{\rm of ~an ~individual ~with ~trait ~\xi ~that ~became ~infected ~\tau ~units ~of ~time ~ago}
     \end{aligned}
     \label{2.1}
\end{equation}

Here $A$ is a measurable non-negative function mapping $\mathbb R_+ \times \Omega \times \Omega$ into $\mathbb R_+$ and $A$ is integrable with respect to $(\tau,\xi)$ over $\mathbb R_+ \times \Omega$. In due time we will make the ‘separable mixing’ assumption that $A$ factorizes as the product of a function of $x$ and a function of $(\tau,\xi)$, but first we shall derive some general results.

Let $S(t,\omega)$ denote the number of susceptible individuals at time $t$ with trait in $\omega$. The model assumes that the measure $S(t,.)$ is absolutely continuous with respect to $\Phi$ with bounded ``derivative''. So for each $t$ a bounded measurable function $s(t,.)$, with $0 \le s(t,x) \le 1$, exists such that

\begin{equation}
    S(t, \omega) = N \int_\omega s(t,x) \Phi(dx).
\end{equation}

Thus $s(t,x)$ can be interpreted as the probability that an individual with trait $x$ is susceptible at time $t$. Let $\Lambda(t,x)$ denote the force of infection at time t on individuals with trait $x$, i.e., assume that

\begin{equation}\label{2.3}
    \partial_t s(t,x) = - \Lambda(t,x) s(t,x).
\end{equation}
    
In doubtful notation we say that the incidence at time $t$ among individuals of trait $\xi$ equals 
\[ - N \partial_t s(t, \xi) \Phi(d \xi).\]

The point is that, directly from the interpretation (2.1) , we have

\begin{equation}\label{2.4}
    \Lambda(t,x) = \int_{0}^{\infty} \int_\Omega  A(\tau,x, \xi) \{- N \partial_t s(t-\tau,\xi) \Phi(d \xi)\} d\tau
\end{equation}

As we show next, we can combine (2.3) and (2.4) to derive a nonlinear RE (Renewal Equation; see \cite{Diekmann2008}) for $s$. To do so, we first integrate (2.3) while assuming that $s(-\infty,x)=1$ for all $x$ (to reflect that a very long time ago, so before the infectious agent made its appearance, the entire population was susceptible):

\begin{equation}
       s(t,x) = \exp \left(-\int_{-\infty}^{t} \Lambda(\sigma,x) d\sigma \right).
\end{equation}

Next we integrate (2.4) with respect to time from $-\infty$ to $t$. Interchanging the order of the integrals, and using that $s$ is a primitive of $\partial_t s$, we arrive at

\begin{equation}
     \int_{-\infty}^{t} \Lambda(\sigma,x) d\sigma  =  N \int_{0}^{\infty} \int_{\Omega}  A(\tau,x,\xi) [ 1 - s(t-\tau,\xi)] \Phi(d\xi) d\tau.
\end{equation}

Upon substitution of (2.6) into (2.5) we obtain the equation

\begin{equation}\label{main_equation}
      s(t,x) = \exp\left( - N \int_{0}^{\infty} \int_\Omega  A(\tau,x,\xi) [ 1 - s(t-\tau,\xi)] \Phi(d\xi) d\tau\right)
\end{equation}
which provides a complete mathematical representation of the model and serves as the starting point of our analysis in the next section.
Here we do not discuss the initial value problem, corresponding to prescribing the history of $s$ on a time interval extending back to $-\infty$, nor the dynamical systems point of view, corresponding to shifting in time along the function $s$ obtained by extending the given history. Instead, we refer to \cite{Diekmann2008}, \cite{Diekmann2007} and \cite{Diekmann2011} for an exposition of the relevant ideas. Also see \cite{Thieme1985}.

In conclusion of this section, we mention that in the recent paper \cite{Diekmann2021}, it is argued that the discrete time variant
\begin{equation} \label{main_equation_discrete}
     s(t,x) = \exp\left( - N \sum_{j=1}^{\infty} \int_\Omega  A_j(x,\xi) [ 1 - s(t-j,\xi)] \Phi(d\xi)\right)
\end{equation}
offers great computational advantages, in particular when $\Omega$ is a finite set. The key point is that there is no user-friendly tool to solve (\ref{main_equation}) and that (\ref{main_equation_discrete}) is straightforward to implement. For the numerical analysis point of view, we refer to Messina et al. \cite{Messina2022a,Messina2022b, Messina2023_a,Messina2023,Messina2023b}.

\section{The linearized problem}
For given $t$, we think of $s(t,\cdot)$ as an element of 
\begin{equation}\label{set_Y}
	\mathcal Y = \{ \psi : \psi~{\rm~ is~ a~ bounded~ measurable~ function~ } \Omega \to \mathbb R \}
\end{equation}
provided with the norm
\begin{equation}
    \| \psi \| = \sup_{x \in \Omega} |\psi(x)|.
\end{equation}
Equation (\ref{main_equation}) admits the disease free steady state solution $s=1$ identically. Inserting
\begin{equation}
   s(t,x) := 1 - y(t,x) 
\end{equation}
into (\ref{main_equation}) and assuming that $y$ is small, we obtain, upon neglecting the higher order terms in the Taylor expansion, the linearized equation
\begin{equation}
    y(t,x) = N \int_{0}^{\infty} \int_{\Omega}  A(\tau,x,\xi) y(t-\tau,\xi) \Phi(d\xi) d\tau.
\end{equation}

We refer to Sections 5 and 6 of \cite{Thieme1984} for an early profound analysis of such linear equations within the setting of positive operator theory.

When, as an Ansatz, we put
\begin{equation}
  y(t,x) = \exp(\lambda t) \psi(x)
\end{equation}
we obtain the nonlinear eigenvalue problem
\begin{equation}
  \psi = \mathcal{K}_\lambda \psi
\end{equation}
with, for $\lambda$ in a right half plane of $\mathbb{C}$,
\begin{equation} \label{K_lambda}
    (\mathcal{K}_\lambda \psi)(x) := N \int_{\Omega} k_\lambda(x,\xi)\psi(\xi)\Phi(d\xi)
\end{equation}
\begin{equation}
    k_\lambda (x,\xi) := \int^{\infty}_0 A(\tau,x,\xi) e^{-\lambda \tau} d\tau.
\end{equation}
We assume that 
\begin{equation}\label{k0_finite}
    \sup_{(x,\xi) \in \Omega \times \Omega} k_0(x,\xi) < \infty
\end{equation}
and interpret $\mathcal K_0: \mathcal Y \to \mathcal Y$ as the Next Generation Operator (NGO), but this needs a bit of explanation.

The standard approach, as presented in \cite{Diekmann2013} and \cite{Inaba2017}, is to define the NGO as a bounded positive operator on $\mathcal L_1(\Omega)$. The generality achieved by describing the population composition by the measure $\Phi$, precludes this in the present situation. We have to work with measures on $\Omega$ that are not necessarily absolutely continuous with respect to the Lebesgue measure. So here, guided by the interpretation just as in the standard approach, we introduce $\mathcal K: M(\Omega) \to M(\Omega)$ by defining 
\begin{equation}
    (\mathcal K m)(\omega) = N \int_{\omega} \Big( \int_\Omega k_0 (x,\xi) m(d\xi) \Big) \Phi(dx).
\end{equation}
%We equip $M(\Omega)$ with the total variation norm and define the Basic Reproduction Number $R_0$ as the spectral radius of $\mathcal K$.

Now note that $\mathcal K$ maps a measure $m$ to a measure of the special form \[\omega \mapsto N \int_\omega \psi(x) \Phi(dx)\] for some $\psi \in \mathcal Y$. The interpretation is that $\psi$ specifies the distribution over $\Omega$ in the form of a fraction. This observation motivates us to define a bounded linear operator $T: \mathcal Y \to M(\Omega)$ by
\begin{equation}\label{T_operator}
    (T \psi)(\omega) = N \int_\omega \psi(x) \Phi(dx).
\end{equation}
We assume that $T$ is \underline{injective} or, in other words, that 
\begin{equation*}
    \int_\omega \psi(x) \Phi(dx) = 0,~ \forall \omega ~\Rightarrow ~\psi(x) = 0, ~\forall x
\end{equation*}
(so, for instance, if $\Omega$ is a finite set, then $\Phi$ should be positive for all points). In this case we have
\begin{equation}\label{K0_fractions}
    \mathcal K_0 = T^{-1} \mathcal K T
\end{equation}
and the conclusion is that $\mathcal K_0$ is indeed the NGO, but represented in terms of \underline{fractions}. We define the Basic Reproduction Number $R_0$ as the spectral radius of $\mathcal K_0$.

In \cite{Andreasen2011}, V. Andreasen presents more or less the same observation for the special case that $\Omega$ is a finite set.

%If $T$ is also surjective, we can conclude from the bounded inverse theorem that $T^{-1}$ is bounded too. But in general $T$ will not be surjective. Also, the range of $T$ need not be closed (for instance, if $\Phi$ is the Lebesgue measure on a nice bounded subset of $\mathbb{R}$, the closure consists of measures of the form of the right hand side of (\ref{T_operator}), with $\psi$ integrable but not necessarily bounded). The point is that we want to verify that the spectral radius of $\mathcal K_0$ is equal to the spectral radius of $\mathcal K$. The following observation seems the easiest way to achieve this. Note that
%\begin{equation}
%    (T^{-1} \mathcal K)(m) = \int_{\Omega} k_0(\cdot,\xi)m(d\xi).
%\end{equation}
%On account of (\ref{k0_finite}) this shows that $T^{-1} \mathcal K: M(\Omega) \to \mathcal Y$ is bounded. Since, according to Gelfand's formula,
%\[\text{spectral radius}~(L) = \lim_{k \to \infty} \| L^k \|^{1/k}\] this observation and (\ref{K0_fractions}) suffice to conclude that $\mathcal K_0$ and $\mathcal K$ do indeed have the same spectral radius.

We define the Malthusian parameter $r$ as the unique real root of 
\begin{equation}
    \text{spectral radius}~\mathcal K_\lambda = 1
\end{equation}
if a real root exists. It is of interest to determine precise conditions on $A$ and $\Phi$ such that
\begin{enumerate}
\item [—] $R_0$ is an eigenvalue with a corresponding positive eigenvector
\item [—] $r$ exists
\item [—] $\mathcal K_r$ has eigenvalue 1 with corresponding positive eigenvector
\item [—] ${\rm sign}(R_0-1) = {\rm sign}( r)$
\end{enumerate}
We refer to \cite{Inaba2017} for results in the $\mathcal L_1(\Omega)$ setting and to \cite{Franco2022} for results in the $M(\Omega)$ setting. When working with the supremum norm (so when investigating $\mathcal K_0: \mathcal Y \to \mathcal Y$) it is helpful to strengthen the conditions on $k_0$ such that the range of $\mathcal K_0$ consists of continuous functions (in particular, in order to use the Arzela-Ascoli compactness criterion). We shall in fact do this in the next section, when discussing the final size as a function of $x$. But we do not elaborate these spectral considerations here, since we want to emphasize that an enormous simplification is achieved by assuming that $A$ is the product of functions of less than three variables. This is elaborated in Section \ref{section_sep1}. In conclusion of this section we refer to \cite{Breda2021, Breda2020, Breda} for recently developed numerical methods for the computation of $R_0$. And to \cite{Thieme2009a} for persistence results for an analogous model that does take demographic turnover into account.

\section{Herd Immunity}

When the outbreak progresses, the size of the susceptible subpopulation declines. At any time \(t\), we can perform a thought experiment: suppose we remove instantaneously and with very high probability every individual that contributes to the future force of infection, does the outbreak make a restart and flare up or does it die out? If it dies out, we say that ‘Herd Immunity’ has been reached. The implication is that the incidence will dwindle, as, on average, new cases generate less than one case (often expressed by saying that the effective reproduction number, denoted by \(R_{\text{eff}}(t)\), is less than one). The implication is NOT that in total there will be only a few future cases. Indeed, in actual fact there is a reservoir of already infected individuals that together generate a considerable force of infection and thus a prolonging of the outbreak and an increase of the outbreak size. When herd immunity is reached by vaccination, the population is safe. When herd immunity is reached during an outbreak, better times loom on the horizon, but the danger has not passed, the overshoot may be substantial (see \cite{Nguyen2023}).

For a homogeneous host population, the situation is simple: herd immunity is reached when the susceptible fraction passes the value \(1/R_0\). The complementary fraction $1 - 1/R_0$ is called the Herd Immunity Threshold, which is usually abbreviated to HIT.
In the SIR compartmental model, this coincides with the prevalence \(I\) and the force of infection \(\beta I\) reaching a maximum. As a consequence, there is a tendency to identify “reaching the peak” and “passing the HIT”. But this is unwarranted, if only since different indicators of “severity” may reach a peak at different moments in time. Indeed, even for the SEIR compartmental model it happens that the peak in the force of infection \(\beta I\) does not coincide with the peak in the prevalence, if one defines prevalence by \(E + I\). See Section \ref{subsect_peaks} below for other examples.

For a heterogeneous host population, the situation is, in general, rather complicated. For any measurable function \(s: \Omega \rightarrow [0,1]\), one can define \(R_{\text{eff}}\) by the methodology described in Section 3 and in \cite{Inaba2023}. The condition \(R_{\text{eff}} = 1\) defines a codimension one manifold in the infinite-dimensional space of bounded measurable functions mapping \(\Omega\) into \([0,1]\), and the point where this manifold is ‘passed’ may very well depend on the initial condition, i.e., on the precise way in which the outbreak is triggered, see Section \ref{subsect_well_hit} below and see \cite{Oliva2023}. As a consequence, the HIT, as the complement of the fraction of the population still susceptible upon reaching herd immunity, is NOT a well defined CONCEPT. (As a side remark we note that for uniform vaccination there is no ambiguity, provided the contact process is in no way influenced by vaccination status; for such a vaccination scenario the HIT is equal to $1 - 1/R_0$, also for a heterogeneous population.)

In Section \ref{section_sep2}, we shall show that the heterogeneous situation becomes as simple as the homogeneous situation if the dynamics can be fully described in terms of a scalar function of time. Following the footsteps \cite{Gomes2022, Montalban2022} of Gabriela Gomes this will allow us to show that heterogeneity can lead to a major “reduction” of the HIT (terminology is dangerous here; what we mean is, that the fraction of the population still susceptible when herd immunity is reached, is substantially higher than \(1/R_0\)). Note that in the very recent paper \cite{Ball2023} an example is presented where heterogeneity, in this case household structure, has the opposite effect!

We close with a word of warning: here we stay within the idealized world of models; using early phase data, to predict when herd immunity will be reached, is afflicted with serious and subtle difficulties, see \cite{Castro2020} and the references given there.

\section{The Final Size Equation}

The interpretation requires that, for fixed $x$, $s(t,x)$ is a monotone non-increasing function of $t$. (Standard arguments can be used to show that, if one prescribes for each $x$ the history of $s$ on an interval of the form $(-\infty, t_0]$ by a non-increasing function with values in $[0,1]$, then the equation defines a unique non-increasing extension to $(-\infty,+\infty)$ with values in $[0,1]$.) As a bounded monotone function must have a limit, we know that $s(\infty,x)$ exists. By passing to the limit in \eqref{main_equation} we obtain the so-called final size equation

\begin{equation}\label{eq_final_size}
     s(\infty,x) = \exp\Big( -\mathcal K_0 \, (1 - s(\infty, \cdot))(x) \Big).
\end{equation} 
Please note the general form of this equation: it only depends on the particular model by way of the representation of the NGO in terms of fractions.
Equation (\ref{eq_final_size}) has $s(\infty, \cdot) \equiv 1$ as trivial solution. This describes the situation in which the pathogen is \textit{not} introduced in the host population. From now on we adopt the hypotheses
\begin{equation}
\tag*{H\textsubscript{A\textsubscript 1}}
\forall \epsilon>0, \exists\delta>0 \mbox{ such that }
\int_{\Omega}\left|k_0(x_1,\xi)-k_0(x_2,\xi)\right|\Phi(d\xi)<\epsilon \mbox{ when } \left|x_1-x_2\right|<\delta\label{4.2}
\end{equation}
and 
\begin{equation}
\tag*{H\textsubscript{A\textsubscript 2}}
	\exists m~{\rm such~that~} \inf_{(x, \xi) \in \Omega \times \Omega} k^n(x, \xi) > 0~{\rm for~} n \geq m \label{4.3}
\end{equation}
where  $k^n$ is  defined inductively by
\begin{equation*}
	\begin{aligned}
		k^1(x, \xi) &:= k_0(x,\xi) \cr
		k^{n+1}(x,\xi) &:= \int_{\Omega} k(x,\eta) k^n(\eta, \xi) \Phi(d\eta),
	\end{aligned}
\end{equation*}
Below we demonstrate that these reasonable conditions on $A$ 
guarantee that 
\begin{itemize}
	\item [-] for $R_0 > 1 $, equation (\ref{eq_final_size}) has precisely one nontrivial solution taking values in $[0,1]$; in fact the values are bounded away from both 0 and 1
	\item [-] for $R_0 \leq 1$  no such nontrivial solution exists.
\end{itemize}
So what happens when $R_0 \leq 1$ and we do infect a very small fraction of the host individuals with the pathogen? As explained in the elaboration of Exercise 1.22.iv in \cite{Diekmann2013}, the final size will be a Lipschitz continuous function of the size of the introduction, hence will be of the same order of magnitude as the introduction. In contrast, when $R_0 > 1$ such an introduction, no matter how small, causes a large outbreak as described by the nontrivial solution of (\ref{eq_final_size}). Here we add  a warning: our deterministic description ignores the demographic stochasticity inherent to small {\it numbers}. We refer to section 1.3.4 in \cite{Diekmann2013} for a description of the effects of demographic stochasticity. 

   The condition \ref{4.2} on $A$ has a double effect: it first guarantees that functions in the range of $\mathcal K_0$ are continuous (a weaker condition would suffice for that) and, next, that the restriction of $\mathcal K_0$ to the continuous functions on $\Omega$ is compact (by the Arzela-Ascoli Theorem) if $\Omega$ is compact. The condition \ref{4.3} guarantees that $\mathcal K_0$ is irreducible in the strong sense that a power of $\mathcal K_0$ maps the positive cone, with the zero element excluded, into the interior of that cone (just like a primitive matrix). These properties of $\mathcal K_0$ will be used in the proofs below. These proofs are inspired by \cite{Schmidt1990, Thieme1979, Thieme1980} and Appendix B in \cite{Rass2003}. For other interesting aspects of the final size equation we refer to \cite{Andreasen2011,Barril2023,Inaba2014,Katriel2012, Brauer2009}.

Defining
\begin{equation}
	y(x) = 1 - s(\infty, x)
\end{equation}
we rewrite (\ref{eq_final_size}) as
\begin{equation}\label{eq_final_size_rew}
	y = F(y)
\end{equation}
where 
\begin{equation}\label{eq_final_size_rew_F}
	F(y)(x) := 1 - e^{-(\mathcal K_0 y)(x)}
\end{equation}
In the following we use ``$\geq$'' to denote the order relation induced by the cone of nonnegative functions on $\Omega$. So 
\begin{equation}
	y_1 \geq y_2 ~ \iff ~ y_1(x)  \geq y_2 (x) {\rm , }~ \forall x \in \Omega {\rm .} 
\end{equation}
$\mathcal K_0$, being a positive linear operator, is \textit{order preserving}:
\begin{equation}\label{eq_K_orderpreserving}
	y_1 \geq y_2 \Rightarrow \mathcal K_0 y_1 \geq \mathcal K_0 y_2
\end{equation}
\begin{lemma} $F$ is order preserving, i.e., 
	\begin{equation}
		y_1 \geq y_2 \Rightarrow F(y_1) \geq F(y_2).
	\end{equation}
\end{lemma}
\begin{proof}[Proof]
Combine (\ref{eq_K_orderpreserving}) with the fact that $z \longmapsto 1 - e^{-z}$ is monotone increasing.
\end{proof}
We want to find solutions $y \geq 0$ of (\ref{eq_final_size_rew}). The form (\ref{eq_final_size_rew_F}) of $F$ implies that any such solution satisfies $y < 1$, in the sense that $y(x)<1$ for all $x$.
\begin{theorem} If $R_0 < 1$, equation (\ref{eq_final_size_rew}) has only the trivial solution $y=0$.
\end{theorem}
\begin{proof}
Since $1 - e^{-z} \leq z$ for $z \geq 0$ (see Lemma \ref{lemmaA1}.i) the inequality $F(y) \leq \mathcal K_0 y$ holds. So (\ref{eq_final_size_rew}) yields $y \leq \mathcal K_0 y$ and, by induction, $y \leq \mathcal K_{0}^n y$ for $n \geq 1$. Taking the supremum with respect to $x$ at both sides, it follows that 
\[\sup_{x \in \Omega} y \leq \| \mathcal K_0^n \| \sup_{x \in \Omega} y \].

Now recall Gelfand's formula for the spectral radius 
\[R_0 = \rho(\mathcal K_0) = \lim_{n \to \infty} \| \mathcal K_0^n \|^{\frac{1}{n}}.\]
Let $n$ be so large that $\| \mathcal K_0^n\|^{\frac{1}{n} }< 1$, then also $\| \mathcal K_0^n \| < 1$ and the inequality above can only hold if $\sup y = 0$.

\end{proof}

\begin{lemma}
Let $y$ be a nontrivial solution of (\ref{eq_final_size_rew}). If \ref{4.3} holds, then $y(x) > 0$ for all $x \in \Omega$.
\end{lemma}
\begin{proof}[Proof]
Define $z_0=\sup\limits_{x\in\Omega}\left(\mathcal K_0y\right)(x)$ and
 use the inequality (ii) of Lemma \ref{lemmaA1} to deduce that
\[y = F(y) \geq (\frac{1 - e^{-z_0}}{z_0}) \mathcal K_0 y.\]
By induction 
\[y \geq \left(\frac{1 - e^{-z_0}}{z_0}\right)^n \mathcal K_0^n y\] and hence \ref{4.3} implies that $y(x)>0$. 
%Now use \textcolor{red}{[3.12]. [DC: The following assumption needs to be stated somewhere (from old Section 3)].
%}
\end{proof}

\begin{theorem}
Let $\Omega$ be compact and assume that \ref{4.2} and \ref{4.3} hold. Then
equation (\ref{eq_final_size_rew}) has at most one nontrivial solution.
\end{theorem}
\begin{proof}[Proof]
Let $y$ and $z$ be nontrivial solutions. Both are continuous  and strictly positive on the compact domain $\Omega$. Define
\[\theta = \min_{x \in \Omega} \frac{y(x)}{z(x)}\]
and assume that $\theta < 1$. Let $\bar{x} \in \Omega$ be such that $y(\bar{x}) = \theta z(\bar{x})$. Now use Lemma \ref{lemmaA1}.iv to obtain the contradiction
\begin{equation*}
	\begin{aligned}
		y(\bar{x}) &= F(y)(\bar{x}) = 1 - e^{- (\mathcal K_0 y)(\bar{x})} \geq 1 - e^{-\theta (\mathcal K_0 z)(\bar{x})} \cr
		& > \theta (1 - e^{-(\mathcal K_0 z)(\bar{x})}) = \theta z(\bar{x}).
	\end{aligned}
\end{equation*}
So $\theta \geq 1$ and $y(x) \geq z(x)$. But if we consider 
\[\tilde{\theta} = \min_{x \in \Omega} \frac{z(x)}{y(x)}\] 
exactly the same argument yields $\tilde{\theta} \geq 1$ and $z(x) \geq y(x)$. We conclude that $y(x) = z(x)$ for all $x \in \Omega$.
\end{proof}

\begin{theorem}
Let $\Omega$ be compact and assume that \ref{4.2} and \ref{4.3} hold.
If $R_0 > 1$, equation (\ref{eq_final_size_rew}) has precisely one nontrivial solution.
\end{theorem}
\begin{proof}Assumption \ref{4.2} guarantees that $\mathcal K_0$ is compact and assumption \ref{4.3} that $\mathcal K_0$ is irreducible. On account of the (sharpening of the) Krein-Rutman Theorem (by de Pagter) we conclude that
 $\mathcal K_0$ has an eigenvector $y_0 \geq 0$ corresponding to $R_0$, normalized by $\max_{x \in \Omega} y_0(x)=1$.  Choose $\epsilon > 0$ small enough to have $(1-\epsilon) R_0 > 1$. Choose $\delta=\delta(\epsilon)$ as in Lemma \ref{lemmaA1}.iii. Let $\tilde{\delta} = \delta / R_0$. Then 
\begin{equation*}
	\begin{aligned}
		F(\tilde{\delta} y_0) &= 1 - e^{-\tilde{\delta} \mathcal K_0 y_0} = 1 - e^{-\delta y_0} \geq (1 - \epsilon) \delta y_0 \cr
		& = (1 - \epsilon) R_0 \tilde{\delta} y_0 \geq \tilde{\delta} y_0.	
	\end{aligned}
\end{equation*}
Iterating $F$, starting with $\tilde{\delta} y_0$, we obtain an increasing and bounded sequence that converges uniformly on $\Omega$ to a solution.
\end{proof}

%[To justify the limit, we probably need that $s(t,x)$ converges uniformly in $x$. As a matter of fact, I have no precise idea about a natural space of functions of $x$ in which $s(t,\cdot)$ lives. Do we want to formulate precise assumptions on $A$ and $\Phi$ that allow us to elaborate the details ?
%Note that $s(\infty,x) = 1$ for all $x$ is a solution of (5.1). Do we want to prove that for $R_0 > 1$ there is a unique nontrivial solution with values in $[0,1]$ ? And that for $R_0 < 1$ and $R_0=1$ there is no non-trivial solution ? In the ageing-in-Hong-Kong manuscript we do so for $\Phi$ a measure supported in finitely many points. I am optimistic that the proof can be extended to the current more general situation. But it does require some work.

  Boundedness of the trait space $\Omega$ is essential for the uniform convergence towards the final situation. Indeed, when the trait specifies the spatial position on the line or in the plane, expansion of a locally introduced infection is characterized by an asymptotic speed of propagation, which is equal to the smallest possible speed $c_0$ of plane wave solutions, see \cite{Rass2003} and the references in there. %(for a review of [29] click 
%%https://www.ams.org/journals/bull/2005-42-04/S0273-0979-05-01061-X/S0273-0979-05-01061-X.pdf ).

 Intuitively, it is crystal clear that the Basic Reproduction Number associated with the situation AFTER the outbreak should be less than one. Remarkably, a simple straightforward proof does not exist, as far as we know. The proof presented below is inspired by the proof given in \cite{Schmidt1990} for the finite dimensional case.
\begin{theorem}
Assume that $\Omega$ is compact, that \ref{4.2} and \ref{4.3} hold and that $R_0 > 1$.  The NGO corresponding to the situation after the outbreak has spectral radius less than one.
\end{theorem}
\begin{proof}
To describe the situation after the outbreak, we replace $\Phi$ by the measure 
\begin{equation}
    \omega \mapsto \int_{\omega} s(\infty,x) \Phi(dx)
\end{equation}
where $s(\infty,.)$ is the unique nontrivial solution of (\ref{eq_final_size}). Introducing
\begin{equation}
y(x):= 1 - s(\infty,x)
\end{equation}
we have
\begin{equation}
 \mathcal K_0^{\mbox{after}} = \mathcal K_0 L
\end{equation}
where $L : \mathcal Y \to \mathcal Y$ is defined by
\begin{equation}
L \psi = (1-y) \psi
\end{equation}
So if we define
\begin{equation}
  \bar{\mathcal K}_0 = L \mathcal{K_0}
\end{equation}
then
\begin{equation}
  \bar{\mathcal K}_0 = L \mathcal K_0^{\mbox{after}} L^{-1}
\end{equation}
We conclude that $ \bar{\mathcal K}_0$  and $\mathcal K_0^{\mbox{after}}$ are similar, consequently have the same spectrum and therefore have the same spectral radius.

From \ref{4.2} it follows that $ \bar{\mathcal K}_0$, as a linear operator on $C(\Omega)$, is compact. Consequently, its adjoint, acting on $M(\Omega)$, is compact too. By the Krein-Rutman Theorem, the spectral radius is an eigenvalue of this adjoint with a nontrivial positive measure $\mu$ as the corresponding eigenvector. We denote the spectral radius by $\rho$.
Now rewrite (\ref{eq_final_size}) as
\begin{equation}
  \mathcal K_0 y + \ln{(1-y)} = 0
\end{equation}
multiply both sides by $1 - y$ and integrate against $\mu$ over $\Omega$. This leads to the identity
\begin{equation}
 N \int_{\Omega} \mu(dx) \left[ \rho y(x) + ( 1 - y(x) ) \ln{( 1 - y(x) )} \right]  = 0
 \label{Nintomega}
\end{equation}
For $\rho \geq 1$ and $0 < z < 1$ the inequality
\begin{equation}
  \rho z + ( 1 - z )\ln{( 1 - z )} > 0
\end{equation}
holds. As a consequence, the left-hand side of (\ref{Nintomega}) is strictly positive for $\rho \geq 1$ and it follows that necessarily $\rho < 1$. 

\end{proof}

\section{Separable Mixing: reduction to a scalar renewal equation}\label{section_sep1}

A preliminary conclusion: type-structure complicates epidemic models, but all the well-known results from the single-type situation do have a multi-type analogon. What we want to know, however, is what impact heterogeneity has on the dynamics. When it comes to doing calculations, the difference between the structured and the unstructured situation can be enormous. 

Following Section 8.4 of \cite{Diekmann2013} we now show how the assumption of separable mixing, (6.1) below, allows us to work with scalar quantities and thus facilitates the computational aspect tremendously. The key point is that various operators have a one-dimensional range. The interpretation is as follows: whenever the type at the moment of becoming infected is following an a priori given distribution (in particular independently of the type of the infecting individual), newly infected individuals are identical in a stochastic sense and therefore we know how to take averages. The aim of this section is to demonstrate that this principle is not restricted to $R_0$, but extends to other aspects of the spread of infectious agents. 

In recent Covid-driven work (\cite{Gomes2022}, \cite{Montalban2022}, \cite{Tkachenko2021}, \cite{Neipel2020}, \cite{Wong2020}, 
also see \cite{Manrubia2020}), 
this feature was very effectively exploited. See \cite{Novozhilov2008} for a pre-Covid description of the main idea.

   The key mathematical assumption is that the kernel $A$ decomposes into two factors, one describing the influence of the type $x$ of the one that may become infected and one describing the influence of the type $\xi$ and the age-since-infection $\tau$ of the potential infector, i.e.,
   
\begin{equation} \label{sep_two}
    A(\tau, x, \xi) = a(x) g(\tau, \xi)
\end{equation}

When (\ref{sep_two}) holds, the operators $\mathcal K_\lambda$ defined in (\ref{K_lambda}) have one-dimensional range spanned by $a$. Consequently

\begin{equation}
    R_0 = N \int_\Omega \left(\int_{0}^{\infty} g(\tau,\eta) d\tau \right) a(\eta) \Phi(d\eta)
\end{equation}
while $r$ is the unique REAL root of the Euler-Lotka equation

\begin{equation}
    1 = N \int_\Omega \left(\int_{0}^{\infty} g(\tau,\eta) \exp( - \lambda \tau) d\tau \right) a(\eta) \Phi(d\eta)
\end{equation}
Moreover, when we introduce the function $w$ of $t$ by putting

\begin{equation}\label{s_w}
     s(t,x) = \exp( - a(x) w(t) )
\end{equation}
then insertion of (\ref{s_w}) into (\ref{main_equation}) yields for $w$ the RE

\begin{equation}\label{wRE}
    w(t) = N \int_{0}^{\infty} \int_\Omega g(\tau,\eta) \{ 1 - \exp( - a(\eta) w(t-\tau)) \} \Phi (d\eta) d\tau
\end{equation}
with corresponding final size equation

\begin{equation}
    w(\infty) = N \int_{\Omega} \left( \int_{[0,\infty)} g(\tau,\eta) d\tau \right) \{ 1 - \exp( - a(\eta) w(\infty)) \} \Phi(d\eta).
\end{equation}

The Renewal Equation \eqref{wRE} is a delay equation, i.e., a rule for extending a function of time towards the future on the basis of the, assumed to be, known past. Solving such equations numerically is not really difficult, but user-friendly software does not exist (but see Messina et al. \cite{Messina2022a,Messina2022b,Messina2023_a,Messina2023,Messina2023b} for promising developments). A recently developed methodology for numerical bifurcation analysis via systematic approximation by ODE is described in \cite{Scarabel2021}. For models that incorporate demographic turnover, see e.g. \cite{Avram2023,Breda2021}, one could use this approach to study the stability of the endemic equilibrium, but for studying an outbreak in a demographically closed population it seems a bit of overkill (although one could, of course, use the ODE to compute an approximation of the solution of the RE). 
In \cite{Diekmann2021}, it is explained how one can formulate a discrete time variant of (\ref{wRE}) with parameters that allow a clear interpretation (which is very helpful when it comes to identification on the basis of data). If $\Omega$ is a discrete set, with points numbered by an index $i$, and time is represented by the integers, corresponding to, say, days, this variant reads

\begin{equation}\label{discrete_omega}
 w(t+1) = \sum_{k=1}^{k_{max}}\sum_i g_{ki} [1 - \exp( - a_i w(t-k)) ] N_i.
 \end{equation}

Here $N_i := N \Phi(i)$, $a_i$ is a measure for the relative susceptibility of type $i$ and $g_{ki}$ describes the expected infectiousness of an individual of type $i$ at day $k$ after becoming infected. Often one is inclined to assume that $g_{ki}$ factors into the product of a function of $k$ and a function of $i$. Numerical implementation of (\ref{discrete_omega}) is straightforward. Also note that (\ref{discrete_omega}) is the separable mixing simplification of (\ref{main_equation_discrete}).

\section{Separable mixing: detailed analysis, with special attention for Gamma distributed traits and the Herd Immunity Threshold}\label{section_sep2}

When we assume that 
\begin{equation}
    g(\tau,\eta) = b(\tau) c(\eta)
\end{equation}
i.e.,
\begin{equation} \label{sep_asum}
    A(\tau,x,\xi)=a(x)b(\tau)c(\xi)
\end{equation}
we can rewrite (\ref{wRE}) in the form
\begin{equation}\label{wPsi}
    w(t) = \int^{\infty}_0 b(\tau) \Psi(w(t-\tau))d\tau
\end{equation}
with $\Psi: \mathbb{R}_+ \to \mathbb{R}_+$ defined by 
\begin{equation}\label{Psic}
    \Psi(w) = N \int_\Omega c(\eta)\Big(1- e^{-a(\eta)w}\Big) \Phi(d\eta).
\end{equation}
So note that the influence of the model ingredients $N$, $c$, $a$, $\Phi$ is captured by $\Psi$, a scalar function of one variable. In the present section we analyse (\ref{wPsi})-(\ref{Psic}). The aim is to deduce epidemiological relevant conclusions that reveal the impact of heterogeneity. As part of our analyses we shall make specific choices for $a$, $c$, $\Phi$. For instance, note that the choice $a \equiv 1$, $c \equiv 1$ captures the homogeneous situation with
\begin{equation}
    \Psi(w) = N\Big(1 - e^{-w}\Big)
\end{equation} In the companion paper \cite{Diekmann2022} we show how a specific choice of $b$ (involving a matrix exponential, cf. \cite{Diekmann2018}) leads to a large class of ODE systems that correspond to heterogeneous versions of familiar compartmental models. 

By linearization we find that 
\begin{equation} \label{R0_sep}
    R_0 = \Psi'(0) \int^{\infty}_0 b(\tau)d\tau
\end{equation}
with
\begin{equation} \label{Phi0}
    \Psi '(0) = N \int_\Omega c(\eta) a(\eta) \Phi(d\eta)
\end{equation}
and that the Euler-Lotka equation reads
\begin{equation}
    1 = \Psi '(0) \int^{\infty}_0 b(\tau) e^{-\lambda \tau} d\tau.
\end{equation}
Consider an ongoing outbreak. We want to quantify the effect of the reduction in the susceptible subpopulation brought about so far by the infection process itself. To do so, we perform a thought experiment: pretend that the reservoir of already infected individuals does not exist and that the current susceptible subpopulation is the entire population in which the pathogen is introduced. We denote by $R_{\text{eff}}$ the corresponding reproduction number. As long as $R_{\text{eff}} > 1$ the outbreak is still in its accelerating phase, but as soon as $R_{\text{eff}} < 1$ deceleration starts (yet many more victims are to be expected, simply since in reality this reservoir of already infected individuals may well be huge). So $R_{\text{eff}} = 1$ characterizes the turning point at which a gradual improvement slowly sets in. The fraction of the population that is still susceptible when $R_{\text{eff}}$ reaches the value 1, is called the Herd Immunity Threshold (HIT).

In order to determine this HIT for the present model, we require that linearization at the value $\bar{w}$ should yield the value $1$ for the corresponding reproduction number (so the root $0$ for the corresponding Euler-Lotka equation). This amounts to \[1 = \Psi '(\bar{w}) \int^{\infty}_0 b(\tau) d\tau\] which we rewrite in the form
\begin{equation} \label{PsiR0}
    \Psi '(\bar{w}) = \frac{\Psi '(0)}{R_0}.
\end{equation}
Since $\Psi ''$ is negative, there exists a unique solution $\bar{w} > 0$ if $R_0 > 1$. The HIT is the fraction of the population that is susceptible when $w$ reaches the value $\bar{w}$, and is given by 
\begin{equation}\label{s_tilde}
    \bar{s} = \int_\Omega e^{-a(x)\bar{w}} \Phi(dx).
\end{equation}

For $t \to \infty$, $w$ tends to $w(\infty)$ characterized by 
\begin{equation}\label{w_infty}
    w(\infty) = \Psi (w(\infty)) \int^{\infty}_0 b(\tau)d\tau = \frac{\Psi(w(\infty))}{\Psi '(0)} R_0
\end{equation}
and the fraction of the population that escapes is accordingly given by 
\begin{equation}
    \tilde{s} = \int_\Omega e^{-a(x) w(\infty)} \Phi(dx).
\end{equation}
Note that (\ref{w_infty}) implies that $\Psi'(w(\infty)) < \frac{\Psi'(0)}{R_0}$ (since $\Psi(y) > y \Psi'(y)$ for $y>0$) and hence that $w(\infty) > \bar{w}$.

When $\Omega$ is a subset of $\mathbb{R}$, it makes sense to organize the parameterization of susceptibility in such a way that 
\begin{equation}\label{a=x}
    a(x) = x.
\end{equation}
(Note that in the decomposition  (\ref{sep_asum}) we may accommodate multiplicative constants in any of the factors. Here we choose to keep the specification of $a$ and $c$ simple.). Concerning the impact of the trait on the infectiousness, we may now contrast (following work of G. Gomes and co-workers \cite{Gomes2022, Montalban2022}) the case
\begin{equation}\label{c=1}
    c(\xi) = 1
\end{equation}
where there is no impact at all, with the case
\begin{equation}\label{c=xi}
    c(\xi) = \xi
\end{equation}
where susceptibility and infectiousness are perfectly correlated. A straightforward computation based on (\ref{R0_sep})-(\ref{Phi0}), shows that in the second case $R_0$ is a factor \[\text{mean} + \frac{\text{variance}}{\text{mean}}\]
bigger than in the first case, where ``mean'' and ``variance'' refer to the distribution of the trait as described by $\Phi$. This is a well-known result, cf. Section 7.4.1 of \cite{Diekmann2013}, leading to the key insight that for sexual activity structured populations the variance contributes to $R_0$. In principle, we can in a similar manner compare the value $\bar{w}$ for the two cases, but since it is hard to say something in general about the solution of (\ref{PsiR0}), we first choose a particular $\Omega$ and $\Phi$ that will enable us to do explicit calculations.

Let \[\Omega = (0,\infty)\]
and let $\Phi$ be the Gamma Distribution with mean $1$ and variance $\frac{1}{p}$, meaning that $\Phi$ has density
\begin{equation}
    x \mapsto \frac{p^p}{\Gamma(p)} x^{p-1} e^{-px}.
\end{equation}
(That the mean equals one is not a loss of generality, it can be achieved by scaling of the variable $x$, and the scaling constant can, as noted before, be incorporated in the factor $b$ of $A$.) As noted in \cite{Neipel2020, Novozhilov2008, Tkachenko2021}, a key feature of the Gamma Distribution is that, when it is reduced according to (\ref{s_w}), with $a$ given by (\ref{a=x}), we obtain again a Gamma Distribution with exactly the same variance, but a reduced mean. Another important feature is that the Laplace Transform is given explicitly by
\begin{equation}
    \bar{\Phi}(\lambda) = \Big(\frac{\lambda}{p} + 1\Big)^{-p}
\end{equation}
which facilitates the computation of moments via derivatives of the Laplace Transform evaluated in $\lambda=0$. In this context, also note that when (\ref{a=x}) holds and $\Omega = (0, \infty)$ we have
\begin{equation}
    \int^{\infty}_0 e^{-x w(t)} \Phi(dx) = \bar{\Phi}(w(t)).
\end{equation}
Moreover, denoting (\ref{c=1}) as ``Case \rm I'' and (\ref{c=xi}) as ``Case \rm {II}'', we have 
\begin{equation}
    \Psi(w) = N 
        \begin{cases} 
            (1 - \bar{\Phi}(w)) & \text{Case \rm I} \\
            (-\bar{\Phi}'(0) + \bar{\Phi}'(w)) & \text{Case \rm {II.}}
        \end{cases}
\end{equation}
So if $\Phi$ is the Gamma Distribution we obtain
\begin{equation}
    \Psi(w) = 
        \begin{cases}
            1 - \Big(\frac{w}{p} + 1\Big)^{-p} & \text{Case \rm I} \\
            1 - \Big(\frac{w}{p} + 1\Big)^{-p-1} & \text{Case \rm {II}}
        \end{cases}
\end{equation}
and hence
\begin{equation}
    \bar{w} = 
        \begin{cases}
            p\Big(R_0^{\frac{1}{p+1}} - 1\Big) & \text{Case \rm I} \\
            p\Big(R_0^{\frac{1}{p+2}} - 1\Big) & \text{Case \rm {II}}
        \end{cases}
\end{equation}
and, according to (\ref{s_tilde}),
\begin{equation}
    \bar{s} = 
        \begin{cases}
            R_0^{-1 + \frac{1}{p+1}} & \text{Case \rm I} \\
            R_0^{-1 + \frac{1}{\frac{1}{2}p+1}} & \text{Case \rm {II}}
        \end{cases}
\end{equation}
These expressions should be contrasted with the HIT \[\bar{s} = R_0^{-1}\] for the homogeneous situation. We see that in both cases heterogeneity with large variance (i.e., small $p$) brings about a substantial reduction of the HIT. The reason is, of course, that among the individuals infected so far the highly susceptible individuals are over-represented. So here we see that letting the outbreak run its course is a far more efficient way of immunizing a population than at random vaccinating individuals, when there is substantial variation in susceptibility. 

This effect plays already a role in the early stage of the outbreak. For small $w$ we have \[\Psi '(w) = \Psi '(0) - N \int_\Omega c(\xi)a^2(\xi)\Phi(d\xi)w + \ldots\] and 
\begin{align*}
    s = \text{fraction susceptible} &= 1 - \int_\Omega a(\xi)\Phi(d\xi)w + \ldots \\ 
                                &= 1 - w + \ldots
\end{align*}
if we normalize $a$ by requiring \[\int_\Omega a(\xi) \Phi(d\xi) = 1.\] It follows that the reduction in reproduction number relates to the reduction in the fraction susceptible according to \[\frac{\Psi '(w)}{\Psi '(0)} = 1 - \theta (1-s) + o(1-s)\]
with 
\begin{align}
    \theta &:= \frac{\int_\Omega c(\xi) a^2(\xi)\Phi(d\xi)}{\int_\Omega c(\xi)a(\xi) \Phi(d\xi)} \nonumber \\ 
           & =
           \begin{cases}
                1 + \frac{1}{p} & \text{Case \rm I} \\
                1 + \frac{2}{p} & \text{Case \rm {II}.}
           \end{cases}
\end{align}
A. Tkachenko e.a. write in \cite{Tkachenko2021}: ``We named the coefficient $\theta$ the \underline{immunity factor} because it quantifies the effect that a gradual build up of population immunity has on the spread of an epidemic''.

In conclusion of this section we refer to our recent paper \cite{Bootsma2023} for an analysis of the effect of mask wearing on HIT and final size. Our study was inspired by \cite{Pastor-Satorras2022}. See \cite{Eikenberry2020, Ngonghala2020, Tian2023} for the wider context.

\section{The influence of population size on the contact process}\label{section_IPS}

   Transmission is superimposed on contact. In the present section we recall some observations made in \cite{Diekmann2013}, in particular in Section 1.3.3 and Chapter 12, concerning the influence of the population size $N$ on the contact intensity. In the next section we shall focus on the influence of population composition as described by $\Phi$.
   For the sake of exposition, consider the homogeneous SIR model. Whether we write
\begin{equation}\label{8.1}
    \frac{dS}{dt} = - \beta S I
\end{equation}
or
\begin{equation}\label{8.2}
    \frac{dS}{dt} = - \frac{\tilde{\beta}}{N} S I
\end{equation}
is irrelevant as long as $N$ is a given constant, since by
\begin{equation}\label{8.3}
     \tilde{\beta} = \beta N
\end{equation}
we can identify the two equations. But what if we want to compare the spread of the same disease in different geographical areas, for instance countries?
   First of all we should ascertain whether the variables are numbers or spatial densities, cf. Section 1.3.5 of \cite{Diekmann2013}. In stochastic models we deal with finite numbers. In deterministic models we  let numbers go to infinity and yet want to work with something finite, such as a spatial density (number/area) or a fraction (of the total). If we introduce
\begin{equation}\label{8.4}
    s = \frac{S}{N},  \quad  i =\frac{I}{N}
    \end{equation}
then \eqref{8.1}  and \eqref{8.2} transform into, respectively
\begin{equation}\label{8.5}
   \frac{ds}{dt} = - \beta N s i
   \end{equation}
\begin{equation}\label{8.6}
\frac{ds}{dt} = - \tilde{\beta} s i.
\end{equation}

If we want to allow for variable $N$, is it more appropriate to consider $\beta$ as constant (i.e., as independent of $N$) or should we consider $\tilde{\beta}$ as constant? 
If we think in terms of spatial densities and aerosol transmission, ‘contacts’ are reminiscent of colliding molecules in a gas and \eqref{8.5} with $\beta$ constant seems most appropriate. For STD’s (Sexually Transmitted Diseases), for mosquito transmission in a host-vector context and for sun-bathing seals,  \cite{Diekmann2013}, Section 1.3.3 and references given there, \eqref{8.6} with $\tilde{\beta}$ constant seems most appropriate. In the first case, the average number of contacts that an individual has per unit of time scales with population size $N$, so is homogeneous of degree 1. In the second case it is homogeneous of degree 0, so independent of $N$.
A somewhat intermediate situation is described in \cite{Heesterbeek1993, Thieme2000} and Section 12.2 of \cite{Diekmann2013}. It is based on the idea that contacts have a certain duration, so take time. As a consequence there is an upper bound for the number of contacts per unit of time. Yet at small values of $N$ the number of contacts per unit of time is proportional to $N$. 

The underlying contact sub-model assumes that individuals can be ‘single’ or ‘paired with another individual’. Let $N$ measure the total number of individuals, $X$ the number of singles and $2P$ the number of individuals that form a pair with another individual (in other words, there are $P$ pairs). Then
\begin{equation}\label{8.7}
    X + 2 P = N.
\end{equation}

Concerning the dynamics, assume that the rate at which a single becomes part of a pair depends on the availability of potential partners in the sense that it is proportional to $X$, with constant $r$. Assume that a pair spontaneously dissolves at rate $s$ (so has an exponentially distributed life time with mean duration $1/s$). Then
\begin{equation}\label{8.8}
\begin{aligned}
            \frac{dX}{dt} &= - r X^2  + 2 s P\cr
                         \frac{dP}{dt}& = \frac{1}{2} r X^2 - s P.
\end{aligned}
\end{equation}
By making use of \eqref{8.7} we reduce to a scalar differential equation and next it is easy to show that the solution converges to the steady state $\bar{P} = \frac{1}{2} \theta \bar{X}^2$ with
\begin{equation}\label{8.9}
  \theta := \frac{r}{s}
  \end{equation} 
and
\begin{equation}\label{8.10}
\bar{X} := \frac{1}{2\theta}\left( \sqrt{ 1 + 4 \theta N } - 1 \right).
\end{equation}

The probability that a randomly chosen individual participates in a pair equals
\begin{equation}\label{8.11}
   C(N) := \frac{2 \bar{P}}{N} = \frac{ 1 + 2 \theta N - \sqrt{ 1 + 4\theta N }}{ 2 \theta N}
\end{equation}
in steady state. The idea is now to assume that 

\begin{itemize}
\item the processes of pair formation and dissolution occur on a much faster time scale than disease transmission

\item disease status has no influence at all on pair formation and separation

\item transmission only occurs within pairs.
\end{itemize}

So at any moment in time a susceptible belongs to a pair with probability $C(N)$ and, if so, its partner is infectious with probability $I/N$ and, if so, there is a certain probability per unit of time, say $\beta$, that transmission occurs. This leads to
\begin{equation}\label{8.12}
\begin{aligned}
\frac{ dS}{dt} &= - \beta S C(N) \frac{I}{N}\cr 
    \frac{ dI}{dt} &=   \beta S C(N)\frac{ I}{N}  - \alpha I
    \end{aligned} 
\end{equation}
in the familiar SIR setting and to
\begin{equation}\label{8.13}
     \frac{dS}{dt} =  S C(N) \frac{1}{N} \int_{0}^{\infty} \beta(\tau) S (\cdot-\tau) d\tau
     \end{equation}
in the general Kermack-McKendrick framework. Note that $C(N) \sim \theta N$ for small $N$ and that 
$C(N) \to 1$ for $N \to \infty$. For a more detailed justification of \eqref{8.12} we refer to \cite{Heesterbeek1993} and Section 12.2 of \cite{Diekmann2013}. As far as we know, the ‘derivation’ of \eqref{8.12} is so far purely formal.
   One of the things we shall consider in the next section, is a multitype version of the pair formation sub-model described above. This too is based on \cite{Heesterbeek1993}.
   We refer to \cite{Toorians2021} for biologically motivated modelling considerations.

\section{The influence of population composition on the contact process}\label{section_exampleIPS}

   The general model ingredient $A(\tau,x,\xi)$ incorporates information about how expected intrinsic infectiousness depends on $\xi$ and $\tau$, how intrinsic susceptibility depends on $x$, but also on how the probability per unit of time for an individual with trait $x$ to have contact with an individual  of type $\xi$ depends on the combination of $x$ and $\xi$. The aim of this section is to describe a kind of catalogue of possibilities for this last aspect, as first presented in the unpublished manuscript \cite{Chan2016}, which is based on \cite{Chan2013}. The work reported in this manuscript originated from a very concrete question: 
how do we extrapolate information about the impact of the H1N1-2009 Influenza outbreak in Hong Kong to a future outbreak of a similar new influenza strain, taking into account the predictable demographic changes, in particular the ageing of the population, i.e., the relative increase of the older part of the population? The manuscript is available upon request to K.M.D. Chan. 

Recently, the question on how the contact structure depends on the size of different age groups popped up in another context, viz., the effectiveness of measures to prevent the spread of SARS-CoV-2: When only vaccinated individuals are allowed access to certain premises like theaters and restaurants, and the vaccination coverage is inhomogeneous among age groups, the age distribution of individuals at the premises will change. Consequently, the contact intensities between age groups will change as well and the question is how to model these changes in the absence of direct observations of the changes in the contact process at those premises \cite{Bootsma2022}. 

   So here we want to extend the considerations of the foregoing section to the multi-type situation. 
   (For a recent overview focusing on compartmental models see \cite{Hill2023}.)
   
   Recall that we use the words ‘type’ and ‘trait’ interchangeably, but have a tendency to use the former when there are finitely many types and the latter when the trait-space might be, or contain, a continuum.  Here we do indeed restrict to finitely many types, partly for technical reasons, partly to avoid modelling difficulties related to how accurately individuals can distinguish one type of individual from another (what is the difference between an individual on its 70th birthday and an individual that had its 70th birthday a fortnight ago?). In particular we specialize \eqref{main_equation} to
\begin{equation}\label{9.1}
   s_i(t) = \exp \left( - \int_{0}^{\infty} \sum_{j=1}^{m} A_{ij}(\tau) [1 - s_j(t-\tau) ] N_j d\tau \right)
\end{equation}
where $x$ and $\xi$ are replaced by integers numbering the $m$ points in the support of $\Phi$ and  
$N_i := N \Phi(i^{\rm th} \text{ point of support})$. The corresponding final size equation reads
\begin{equation}\label{9.2}
   s_i(\infty) = \exp \left( - \sum_{j=1}^{m}  \int_{0}^{\infty} A_{ij}(\tau) d\tau [1 - s_j(\infty)] N_j \right).
\end{equation}
Next, assume that
\begin{equation}\label{9.3}
    A_{ij}(\tau) = \frac{1}{ N_i}  k_{ij}  b_{ij}(\tau)
    \end{equation}
where
\begin{equation}\label{9.4}
\begin{aligned}
    k_{ij} :=&{\rm expected~ number~ of~ contacts~ per~ unit~ of~ time~ that~}\cr
     &{\rm a~ type-}j~{\rm individual~ has~  with~ type-}i ~{\rm individuals}
    \end{aligned}
    \end{equation}
(so the factor $1/N_i$ serves to translate to a ‘per $i$-type individual’ probability per unit of time) and $b_{ij}(\tau)$ specifies the product of the infectiousness of an $(j, \tau)$ individual and the susceptibility of an $i$-type individual (it is tempting to put $b_{ij}(\tau) = a_i \tilde{b}_j(\tau)$; but at this point we would like to include situations in which, for instance, an individual that is aware of its Covid-19 infection status might choose to care for its children, while avoiding to meet its parents; also note that the factors in a product are never unique, so we are free to put such a reduction of contact intensity into $b$, even though the description in words might suggest to put it into $k$).
   Clearly the consistency relation
\begin{equation}\label{9.5} 
  k_{ij} N_j  =  k_{ji} N_i
\end{equation}
should hold. Let $K$ denote the matrix with elements $k_{ij}$. We allow $K$ to depend on the vector
\begin{equation}\label{9.6}
    {\bf N} = (N_1, \ldots, N_m).
\end{equation}

We call a specification of how $K$ depends on ${\bf N}$ a CONTACT PATTERN.

In general, a contact pattern $K$ can be represented by a function  $K: \mathbb{R}^m\to\mathbb{R}^{\frac{m(m+1)}{2}}$, as relation (\ref{9.5}) implies that the upper-triangular part of the matrix $K$ specifies the full matrix $K$. Moreover, one expects a contact pattern to be continuous, such that small changes in  {\bf N} lead to small changes in $K$.

   When $K$ is homogeneous of degree 1 and, more precisely, when
\begin{equation}\label{9.7}
  k_{ij}  =  q_{ij} N_{i} 
  \end{equation}
with $q_{ij} = q_{ji}$ and $q_{ij}$ constant, i.e., independent of ${\bf N}$ for all $1\leq i,j\leq m$,
we call the contact pattern DENSITY DEPENDENT. When, on the other hand, $K$ is homogeneous of degree 0, i.e., when for all $c > 0$ we have
\begin{equation}\label{9.8}
   K(c{\bf N}) = K({\bf N})
   \end{equation}
we call the contact pattern FREQUENCY DEPENDENT.
Note that in the multi-type frequency dependent case, knowledge of $K$ for a certain ${\bf N}$, does not fully specify the contact pattern. Only when all group sizes change with the same factor, the contact matrix $K$ will remain identical. 

In the special case
\begin{equation}\label{9.9}
  k_{ij}  =  \frac{c_i c_j N_i}{\sum_{\ell=1}^{m} c_{\ell} N_{\ell}}
\end{equation}
(with $c_i$ constant) we speak about PROPORTIONATE MIXING, while if for all pairs $(i,j)$ with $i\neq j$ $k_{ij}$ depends on $N_i$ and $N_j$ but NOT on $N_\ell$ for $\ell \neq i, j$ we speak about a BILATERAL pattern. In a bilateral pattern each pair of types of individuals `decides' on how changes in the group size of the two types affect their contact intensities. This `decision process' may differ for each pair of types. Hence, there exist many bilateral patterns. There are two mathematically convenient ways to `decide' on changes in the contact intensities between two different groups $i\neq j$.

One is the POWER LAW\begin{equation}\label{9.10}
   k_{ij} = q_{ij} \left(\frac{N_i}{N_j}\right)^d,
   \end{equation}
which is homogeneous of degree 0, but satisfies the consistency condition \eqref{9.5} only when $d=1/2$.
The other is when there is a DOMINATING type which keeps its contact rate constant, while the other type has to adapt its contact rate to satisfy the consistency condition \eqref{9.5}, i.e., when type $j$ dominates type $i$ we have
\begin{equation}\label{DOMINATING}
\begin{array}{c}
   k_{ij} = q_{ij} \\
   k_{ji} = q_{ij}\frac{N_j}{N_i}.
   \end{array}
   \end{equation}

Note that neither the power-law nor the 
dominating pattern does nail down  the within-group contact intensities $k_{jj}$, even in the case of only two types. 
One may, of course, assume that the within-group contact intensities $k_{jj}$ do not depend on ${\bf N}$, such that equations (\ref{9.10}) and (\ref{DOMINATING}) are valid as well for $i=j$. But other assumptions make sense too, e.g., that the total contact intensity of a type-$j$-individual, $\sum_{i=1}^mk_{ij}$, does not depend on ${\bf N}$.  

% for example, that type-$j$ individuals have a prefered total number of contacts per time unit $K_j$, and we define 
%\begin{equation}\label{kjj}
%k_{jj}=\max{\left(0,K_j-\sum\limits_{i\neq j}k_{ij}\right)}.
%\end{equation}
 
In general, contact patterns are not bilateral.
In practical applications, the contact matrix $K$ is often estimated for a given $\mathbf{N}$, and instead of attempting to determine the full contact pattern, one would like to know the contact matrix $\tilde{K}$ for a specific $\mathbf{\tilde{N}}\neq \mathbf{N}$.
As a contact matrix has $\frac{m(m+1)}{2}$ degrees of freedom, it is a challenge to avoid arbitrariness.

%If for a given $\mathbf{N}$ the  i.e., for the current $\mathbf{N}$ we know the contact rates,   %as equation \eqref{kjj} already suggests. 
One way to deal with the contact pattern problem is to define an explicit model for pair formation, see Hadeler \cite{Hadeler2012}.

   In the Heesterbeek-Metz approach \cite{Heesterbeek1993}, it is assumed that pair formation occurs according to the law of mass action and that pairs disband at a certain rate (or, in other words, pairs exist for an exponentially distributed amount of time). The short time scale dynamic equations for pair formation and dissolution are
\begin{equation}\label{9.11}
\begin{aligned}
            \frac{dX_i}{dt}  &= - \left(\sum_{j=1}^{m} r_{ij} X_j \right) X_i  + 2 \sum_{j=1}^m s_{ij} P_{ij} \cr
            \frac{dP_{ij}}{dt} &=  \frac{1}{2} r_{ij} X_i X_j  -  s_{ij} P_{ij}
            \end{aligned}
            \end{equation}

where we use a similar notation as in \eqref{8.8}, but now with indices indicating the type or the two types forming a pair.

Here we assume that pairs $P_{ij}$ are symmetric entities and that accordingly
\begin{equation}\label{9.12}
     r_{ij} = r_{ji}   ~{\rm and}~   s_{ij}  =  s_{ji}  ~{\rm and}~  P_{ij} = P_{ji}
     \end{equation}
in the sense that the first two identities are requirements for the ingredients $r$ and $s$ while the third is, subsequently, a consequence of \eqref{9.11}. The factor $1/2$ in front of $r_{ij}$ serves to be able to treat $i=j$ and $i \neq j$ in an identical way. The consequence is that the number of pairs consisting of an $i$-type individual and a $ j$-type individual equals $2 P_{ij}$ if $i \neq j$ (or, if you prefer, equals $P_{ij} + P_{ji}$). Accordingly we have
\begin{equation}\label{9.13}
   X_i + 2 \sum_{j=1}^{m} P_{ij} = N_i
   \end{equation}
It can be shown that convergence to a (unique) steady state is guaranteed, see \cite{Heesterbeek1993} and references in there.

To facilitate the notation, we now omit bars when we consider variables in steady state. In steady state we have to have
\begin{equation}\label{9.14}
   P_{ij} = \frac{1}{2} \theta_{ij} X_i X_j
   \end{equation}
with
\begin{equation}\label{9.15}
    \theta_{ij} := \frac{r_{ij}}{s_{ij}}
    \end{equation}

So we can rewrite \eqref{9.13} as
\begin{equation}\label{9.16}
   X_i + \sum_{j=1}^{m} \theta_{ij} X_i X_j  =  N_i.
   \end{equation}
(Incidentally, note that if we divide this identity by $N_i$, the first term at the left hand side is the probability that a type-$i$ individual is single, while the term with index $j$ in the sum gives the probability that a type-$i$ individual is paired to a type-$j$ individual. This observation shall be used below.)

We would like to relate the steady state to a known contact matrix $K$.
As the ${ij}^{\mbox{th}}$ element of the contact matrix $K$ represents the number of contacts a type $j$ individual has, per unit of time, with type $i$ individuals,
we want to find values of $r_{ij}$ and $s_{ij}$ such that 
\begin{equation}\label{eqK}
k_{ij}=2 P_{ij}s_{ij},
\end{equation} since $2P_{ij}s_{ij}$ is the number of $ij$-pairs that dissolve  per time unit, which, in the equilibrium, equals the number of new contacts per time unit. 
 By the symmetry postulated in equation (\ref{9.12}), both $(r_{ij})_{1\leq i,j\leq m}$ and $(s_{ij})_{1\leq i,j\leq m}$ have $m(m+1)/2$ degrees of freedom. So in total there are $m(m+1)$
degrees of freedom for $(r_{ij})_{1\leq i,j\leq m}$ and $(s_{ij})_{1\leq i,j\leq m}$ combined.
As a contact matrix $K$ is determined by its upper triangular part, \eqref{eqK} puts $m(m+1)/2$ restrictions on $s_{ij}$ and $r_{ij}$.
This means that for a general contact matrix $K$, there is a $m(m+1)/2$-dimensional set of $(r_{ij})_{1\leq i,j\leq m}$ and $(s_{ij})_{1\leq i,j\leq m}$ such that the corresponding equilibrium contact process leads to the contact matrix $K$. Hence, additional restrictions on $(r_{ij})_{1\leq i,j\leq m}$ and $(s_{ij})_{1\leq i,j\leq m}$ are needed to uniquely determine how the $r$, $s$ coefficients, and thus how the elements of the contact matrix, change if $\bf{N}$ changes. Here we consider two options.

  In the first option, we simplify the situation. We assume that the two individuals involved have an independent influence on both pair formation and separation, in the sense that both $r_{ij}$ and $s_{ij}$ are the product of an $i$-dependent factor and a $j$-dependent factor, leading to
\begin{equation}\label{9.17}
   \theta_{ij} = \rho_i \rho_j.
   \end{equation}

If we insert
\begin{equation}\label{9.18}
   X_i = \zeta_i N_i
   \end{equation} 
into \eqref{9.16}, use \eqref{9.17} and divide the identity by $N_i$ we obtain
\begin{equation}\label{9.19}
   \zeta_i +\sum_{j=1}^{m} \rho_i \rho_j \zeta_i \zeta_j N_j  = 1
   \end{equation}
which we rearrange into
\begin{equation}\label{9.20}
  \sum_{j=1}^{m} \rho_j \zeta_j N_j = \frac{ 1-\zeta_i}{\rho_i \zeta_i}.
  \end{equation}

As the left hand side does not depend on i, the same must hold for the right hand side. Let us call the common value Q. Then 
\begin{equation}\label{9.21}
   \zeta_i = \frac{1}{ 1 + \rho_i Q}
   \end{equation}
while $Q$ itself should satisfy the equation
\begin{equation}\label{9.22}
  \sum_{j=1}^{m} \frac{\rho_j N_j}{ 1 + \rho_j Q} = Q
  \end{equation}
which has, as a simple graphical argument shows, a unique positive solution. Thus we have constructively defined the solution of the steady state problem for any given $(N_1, \cdots, N_m)$.
   Now how do we use these steady state expressions in the context of an epidemic model? The probability that an $i$-type individual is paired to $j$-type individual is given by
\begin{equation}\label{9.23}
   C_{ij} = \frac{2 P_{ij}}{N_i} = \frac{\rho_i \rho_j X_i X_j}{ N_i}.
   \end{equation}

Under the assumptions formulated below \eqref{8.11} and focussing on an SIR or SEIR setting we should put 
\begin{equation}\label{9.24}
\frac{dS_i}{dt} = - S_i \sum_{j=1}^{m} \beta_{ij} C_{ij} \frac{I_j}{N_j}
\end{equation}
while the analogue of \eqref{8.13} reads
\begin{equation}\label{9.25}
\frac{dS_i}{dt} = - S_i \sum_{j=1}^{m} C_{ij} \frac{1}{N_j} \int_{0}^{\infty} \beta_{ij}(\tau) \frac{dS_j}{dt} (\cdot-\tau) d\tau.
\end{equation}

In \cite{Heesterbeek1993} and the references given there, one finds far more information about how to prove the existence, uniqueness and global stability of the steady state of \eqref{9.11} even when the simplifying assumption \eqref{9.17} does not hold.

In the second option, we do require that the total number of contacts per time unit of an individual does not depend on $\bf{N}$. This idea is inspired by the observation that during the SARS-CoV-2 outbreak, people visiting certain premises, did not change their behaviour substantially \cite{slimopen}. In this context, the type specifies the age group to which an individual belongs.
 We denote the original, known, contact matrix by $K$ and the original population size by ${\bf N}$. By putting tildes on top of these parameters we denote the new situation for which we want to determine the contact matrix $\tilde{K}$. Let $\tilde{N}_j:=\rho_jN_j$, i.e., the lower $\rho_j$, the lower the attendance of type-$j$-individuals in the new situation.

We define an ordering  $\sigma$ such that $\rho_{\sigma(1)}\leq \rho_{\sigma(2)}\leq \ldots\leq \rho_{\sigma(m)}$, so type-$\sigma(1)$ has the highest relative reduction in participation. 

Originally, the average total number of contacts per unit of time of an individual of age group $\sigma(1)$ equalled $\sum_{i=1}^m k_{i\sigma(1)}$. We assume that $\tilde{k}_{j\sigma(1)}$, the number of contacts of a  type-$\sigma(1)$-individual with type-$j$ individuals per unit of time after the intervention, equals:

\begin{equation}\tilde{k}_{j\sigma(1)}=
\frac{\rho_jk_{j\sigma(1)}}{\sum_{i=1}^m \rho_i k_{i\sigma(1)}}\sum_{i=1}^m k_{i\sigma(1)}\label{8.28}\end{equation}
In this way, the total number of contacts of a type-$\sigma(1)$-individual remains $\sum_{i=1}^mk_{i\sigma(1)}$ and the intensity of contacts with age group $j$ are proportional to $\rho_j$ and $k_{j\sigma(1)}$.

To keep contacts symmetric, we need that $\tilde{k}_{\sigma(1)j}$, the number of contacts per unit of time of a type-$j$-individual with type-$\sigma(1)$ individuals, equals:

\begin{equation}
\tilde{k}_{\sigma(1)j}=\frac{k_{\sigma(1)j}}{k_{j\sigma(1)}}
\frac{\rho_{\sigma(1)}}{\rho_j}\tilde{k}_{j\sigma(1)}
%
%
%\tilde{k}_{\sigma(1)j}=
%k_{\sigma(1)j}\rho_{\sigma(1)}\frac{\sum_{i=1}^m %k_{i\sigma(1)}}{\sum_{i=1}^m \rho_i k_{i\sigma(1)}}
.\end{equation}

We have constructed the contact of and with  group $\sigma(1)$ which has the highest reduction in attendance. Next, we will define the contact rate of and with group $\sigma(2)$ individuals. However, $\tilde{k}_{\sigma(1)\sigma(2)}$ and $\tilde{k}_{\sigma(2)\sigma(1)}$ are already defined. As the total contact rate of each type of individual remains constant, the total contact rate of type-$\sigma(2)$-individual with types other than $\sigma(1)$ needs to be:
\begin{equation}
R_{\sigma(2)}:=\left(\sum\limits_{i=1}^mk_{i\sigma(2)}\right)-\tilde{k}_{\sigma(1)\sigma(2)}
 \end{equation}   
We distribute the remaining contact rate $R_{\sigma(2)}$ of type $\sigma(2)$-individuals 
over all types $j\neq \sigma(1)$ in a similar way as we did in \eqref{8.28}, i.e., 
\begin{equation}
   \tilde{k}_{j\sigma(2)}:=
   \frac{\rho_jk_{j\sigma(2)}}{\sum\limits_{i=2}^{m}\rho_{\sigma(i)}k_{\sigma(i)\sigma(2)}}R_{\sigma(2)}
\end{equation}
To keep contacts symmetric, we need that $\tilde{k}_{\sigma(2)j}$, the number of contacts per unit of time of a type-$j\neq\sigma(1)$-individual with type-$\sigma(2)$ individuals, equals:

\begin{equation}
\tilde{k}_{\sigma(2)j}=\frac{k_{\sigma(2)j}}{k_{j\sigma(2)}}
\frac{\rho_{\sigma(2)}}{\rho_j}\tilde{k}_{j\sigma(2)}
%
%
%\tilde{k}_{\sigma(1)j}=
%k_{\sigma(1)j}\rho_{\sigma(1)}\frac{\sum_{i=1}^m %k_{i\sigma(1)}}{\sum_{i=1}^m \rho_i k_{i\sigma(1)}}
.\end{equation}

 We now recursively define all contact rates this way. More precisely,
suppose we know the new contact rate of the $n-1$ types with the highest reduction in attendance, i.e., we know $\tilde{k}_{\sigma(i)j}$ and $\tilde{k}_{j \sigma(i)}$ for $1\leq i\leq n-1<m$ and $1\leq j\leq m$. 
%We know focus on age group $\sigma(2)$
The total contact rate of type-$\sigma(n)$-individuals with all type $\sigma(j)$-individuals with $n-1<j\leq m$ equals:
\begin{equation}
R_{\sigma(n)}:=\left(\sum\limits_{i=1}^mk_{i\sigma(n)}\right)-\sum\limits_{i=1}^{n-1}\tilde{k}_{\sigma(i)\sigma(n)}
 \end{equation}

We define the contact rate of a type $n$-individual with type  $\sigma(j)$ indiviudal  , with $n\leq j\leq m$, as;

\begin{equation}
\tilde{k}_{\sigma(j)\sigma(n)}:=\frac{k_{\sigma(j)\sigma(n)}\rho_{\sigma(j)}}
     {
      \sum_{i=n}^m \rho_{\sigma(i)} k_{\sigma(i)\sigma(n)}
     }R_{\sigma(n)}
\end{equation}
i.e., contacts with types $\sigma(1),\ldots,\sigma(n-1)$ are already defined, and the contacts with the remaining 
types are such that they are proportional to both the original contact rate with that type and the reduction factor of that type. The contacts  are scaled such that the total number of contacts of individuals of age-group $\sigma(n)$ is the same as in the original contact matrix. 
By symmetry of the contacts we have that:

\begin{equation}
\tilde{k}_{\sigma(n)\sigma(j)}=
\frac{k_{\sigma(n)\sigma(j)}}{k_{\sigma(j)\sigma(n)}}
\frac{\rho_{\sigma(n)}}{\rho_{\sigma(j)}}\tilde{k}_{\sigma(j)\sigma(n)}.
%
%
%k_{\sigma(n)\sigma(j)}\rho_{\sigma(n)}
%\frac{
 %      \left(\sum_{i=1}^m k_{i\sigma(n)}\right)-\sum\limits_{i=1}^{n-1}\tilde{k}_{\sigma(i)\sigma(n)}
  %   }
   %  {
    %  \sum_{i=n}^m \rho_{\sigma(i)} k_{\sigma(i)\sigma(n)}
     %}
\end{equation}

Thus we recursively construct the contact rates between all types. This provides us with a new contact matrix ${\bf \tilde{K}}$ which still satisfies the symmetry-condition (\ref{9.5}) and keeps the total contact rate of individuals fixed. This iterative procedure was used in a report for the Dutch government to assess the effectiveness of Corona ticket measures \cite{Bootsma2022}.

Note that these two options do \underline{not} at all exhaust 
all possibilities!

\section{Numerical illustration of some subtle issues}
\subsection{Peaks} \label{subsect_peaks}

   Even without heterogeneity, one needs numerical methods to determine the size and timing of a peak in the incidence and to investigate how these quantities depend on the model ingredients, see e.g. \cite{Diekmann2021}.  With heterogeneity, the need for numerical methods intensifies. And, more importantly, new phenomena arise.
   First, before we can even look for peaks, we need to agree upon the quantity that we graph as a function of time. Here we choose two measures of the incidence, i.e., the number of cases per time unit that become infected and the number of cases per time unit that become infectious. Depending on whether individuals in the latent period (before they become infectious) are symptomatic or not, the first or the latter may be closest to surveillance data in a situation where one may be unaware of relevant heterogeneity.
   Now imagine, as a thought experiment, two well-mixed subpopulations characterized by a major difference of the latent period (reflecting, for instance, a genetically determined difference in immune physiology).
   \begin{figure}[htb]
    \centering
     \begin{subfigure}[b]{0.325\textwidth}
         \centering
         \includegraphics[width=\textwidth]
         {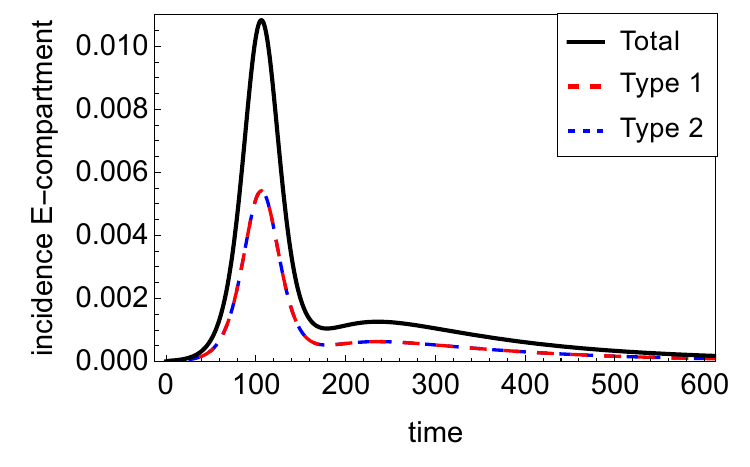}
         \caption{ Incidence of infections}
         \label{exheterolatentinc}
     \end{subfigure}
     \hfill
     \begin{subfigure}[b]{0.325\textwidth}
         \centering
         \includegraphics[width=\textwidth]
         {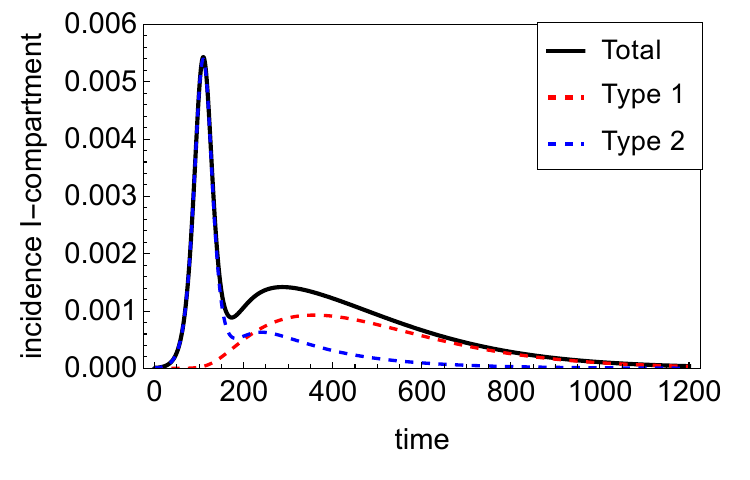}
         \caption{Incidence of infectious cases}  
          \label{exheterolatentincI}
     \end{subfigure}
     \hfill
     \begin{subfigure}[b]{0.325\textwidth}
         \centering
         \includegraphics[width=\textwidth]
         {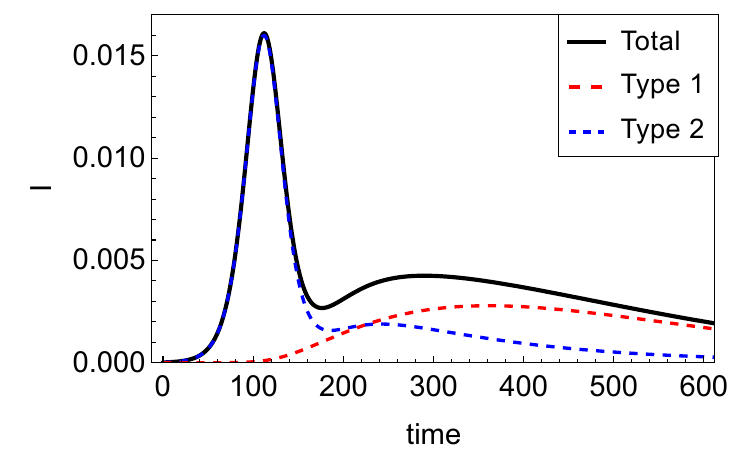}
         \caption{Infectious cases}   
          \label{exheterolatentI}
     \end{subfigure}
   \caption{{\bf Incidence with heterogeneity in latent period.}
We model two groups of equal size and random mixing of both groups. The groups only differ in the duration of the latent period. Both are  Gamma-distributed with shape-parameter 3 while the scale parameter is 1 for group 1 and 0.01 for group 2. The infectious period is exponentially distributed with mean 3 for both groups. $R_0$ equals 3.    (a) Incidence of new infections. (b) Incidence of new infectious individuals. (c) I-compartment as fraction of the total population}\label{examplehetero}
   \end{figure}
   Figure \ref{exheterolatentinc} shows TWO peaks in the graph of total incidence of new infections as a function of time. If we graph the incidence of new infections in the two subpopulations separately, the two graph are identical and the peaks occur at the same points in time for the two subpopulations. But if we graph the incidence of infectious cases Figure \ref{exheterolatentincI} or the contribution of the two subpopulations to the force of infection, Figure \ref{exheterolatentI}, i.e., the number of individuals that are infectious, a different picture emerges: these contributions are out of phase. It is this difference in phase that causes the double peak.
   
   In the example above, contact is uniform but physiology is not. Let's now reverse this and consider neighbouring countries inhabited by identical individuals, but with weak coupling (in the sense of contacts). More precisely, assume that the two subpopulations are of equal size and have equal within-subpopulation contact rate. We introduce the between-subpopulation contact rate as a small parameter $\epsilon$. For $\epsilon=0$ the two subpopulations are uncoupled, each shows a single peak, but the timing of the peak depends, of course, on the timing of the introduction of the pathogen. Now make $\epsilon$ a tiny bit positive. Technically we obtain irreducibility: when we introduce the pathogen in one of the two subpopulations, ultimately both will be hit and, in terms of final size, in equal measure. Yet the outbreaks are bound to be out of phase. 
   \begin{figure}[htb]
    \centering
     \begin{subfigure}[b]{0.325\textwidth}
         \centering
         \includegraphics[width=\textwidth]
         {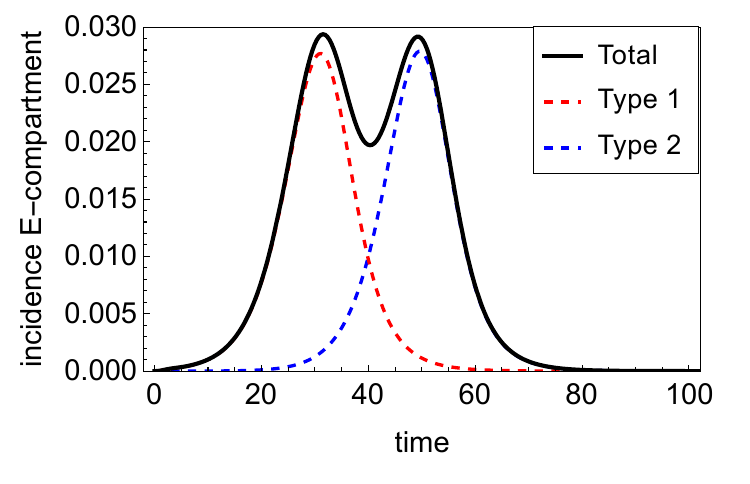}
         \caption{ Incidence of infections}
         \label{incidenceEcoupling}
     \end{subfigure}
     \hfill
     \begin{subfigure}[b]{0.325\textwidth}
         \centering
         \includegraphics[width=\textwidth]
         {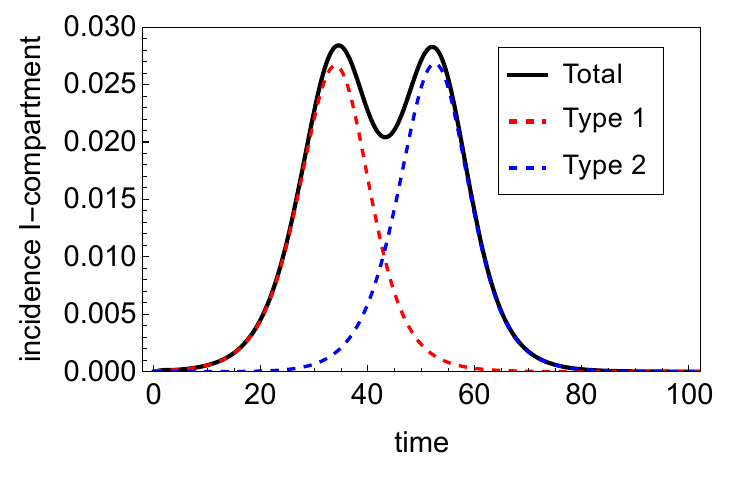}
         \caption{Incidence of infectious cases}  
          \label{incidenceIcoupling}
     \end{subfigure}
     \hfill
     \begin{subfigure}[b]{0.325\textwidth}
         \centering
         \includegraphics[width=\textwidth]
         {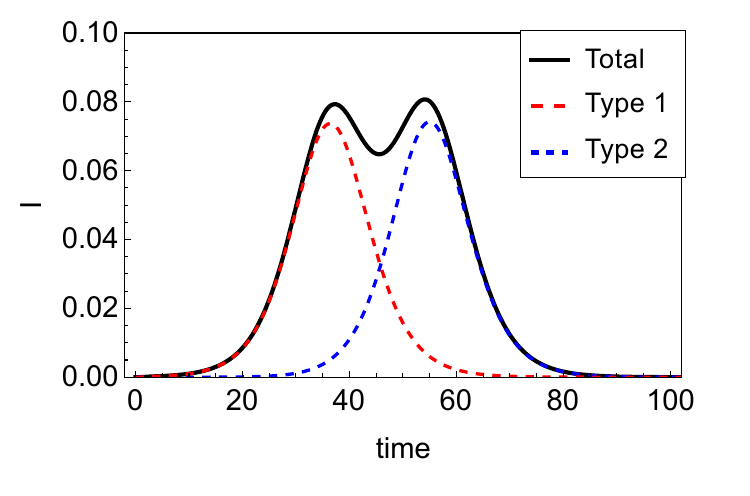}
         \caption{Infectious cases}   
          \label{exheterocouplingI}
     \end{subfigure}
   \caption{{\bf Incidence with heterogeneity due to weak coupling.}
We model two groups of equal size and with weak coupling between the two groups: The within group transmission parameter is 100 times higher than the between-group transmission parameter. The groups have the same parameters (latent period  Gamma-distributed with shape-parameter 3 and scale parameter  1; infectious period exponentially distributed with mean 3 for both groups), $R_0$ equals 3. We introduce the disease one of the groups.  
 (a) Incidence of new infections. (b) Incidence of new infectious individuals. (c) I-compartment as fraction of the total population}\label{exampletiming}
   \end{figure}

   And indeed, in Figure \ref{exampletiming} we can observe TWO peaks in the graphs of the total incidence (both for new infections and for new infectious cases) as a function of time. This time (and in contrast with the first example) each peak relates to the incidence in a subpopulation. Thus, this example illustrates that the question ``when do we speak about one population and when about two ?" has a subtle quantitative aspect, in addition to the more obvious public health administration aspect. At the end of Section 7.3 (page 175) in \cite{Diekmann2013} this is described as `quantitative aspects of irreducibility': it may happen that nonlinear dynamics sets in BEFORE the distribution takes the shape predicted by the stable distribution of the linearized model (this happened for instance with HIV; the gay community suffered considerably before the disease was observable in the heterosexual community, simply since the connection, though non-zero, is so very weak). Or, in other words, the basic reproduction number $R_0$ of the coupled population definitely has critical value one, and yet the linearization might give a prediction/suggestion that isn’t necessarily right.
   In conclusion:
whether a peak occurs also depends on what you plot and, 
even without control measures or arrival of new variants, multiple peaks may be found

\subsection{How well defined is the HIT?}\label{subsect_well_hit}

   As shown in \cite{Inaba2017} and \cite{Franco2022}, the $\infty$-dimensional (i.e., Krein-Rutman) version of Perron-Frobenius theory yields that within the positive cone there is, for the linearized problem, modulo translation only one solution growing away from the disease free steady state. Unstable manifold theory, see \cite{Diekmann2008,Diekmann2011}, next extends this uniqueness modulo translation to the nonlinear setting. Implicitly we have used this idea when writing \eqref{main_equation} and paying little to no attention to the initial condition. Even though we do not know a published proof, we believe that one can define the HIT unambiguously in terms of the positive unstable manifold of the disease free steady state. But, armed with the insights obtained in the preceding subsection, we might wonder how relevant the HIT thus defined is when coupling is only weak?
   To find out, imagine a small subpopulation with a high within-contact rate, very weakly coupled to a large subpopulation that has a within-contact rate large enough to have the within-$R_0$ bigger than 1. 
If we introduce the pathogen in the small subpopulation, the HIT will be reached more or less when the susceptible fraction of the large subpopulation reaches $1/\mbox{within-}R_0$. But if we introduce the pathogen in the large subpopulation, then upon reaching $1/\mbox{within-}R_0$ the small subpopulation will still have its susceptible fraction ABOVE its own 1/within-$R_0$. While we wait for the susceptible fraction of the small subpopulation to reach this critical level, an overshoot happens in the large subpopulation. So we expect that in this second scenario the overall susceptible fraction upon reaching the overall HIT is smaller than in the first scenario. Figure 
\ref{hitstart} illustrates that this can indeed happen.
In conclusion:
for the same model, the HIT may differ substantially depending on the precise details of the introduction, in particular the subpopulation in which the small introduction occurs (we reiterate: the distinction between one population and several populations is not as clear cut as the mathematical definition of irreducibility seems to suggest at first).

  \begin{figure}[htb]
    \centering
     \begin{subfigure}[b]{0.325\textwidth}
         \centering
         \includegraphics[width=\textwidth]
         {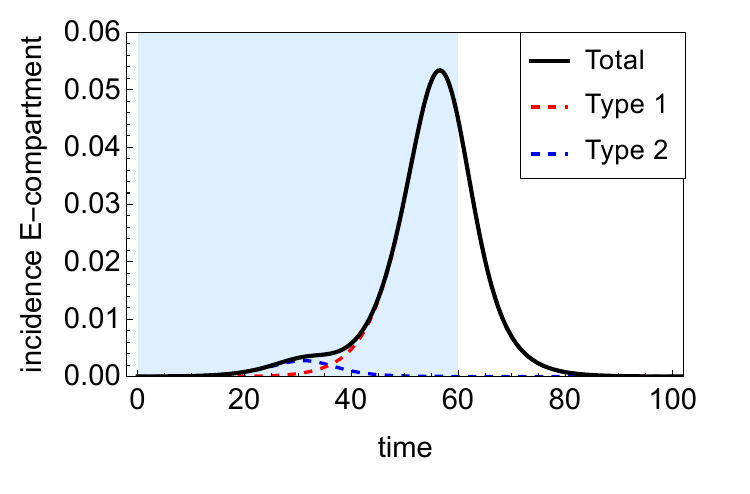}
         \caption{ Incidence of infections}
         \label{incidenceEhitsmall}
     \end{subfigure}
     \hfill
     \begin{subfigure}[b]{0.325\textwidth}
         \centering
         \includegraphics[width=\textwidth]
         {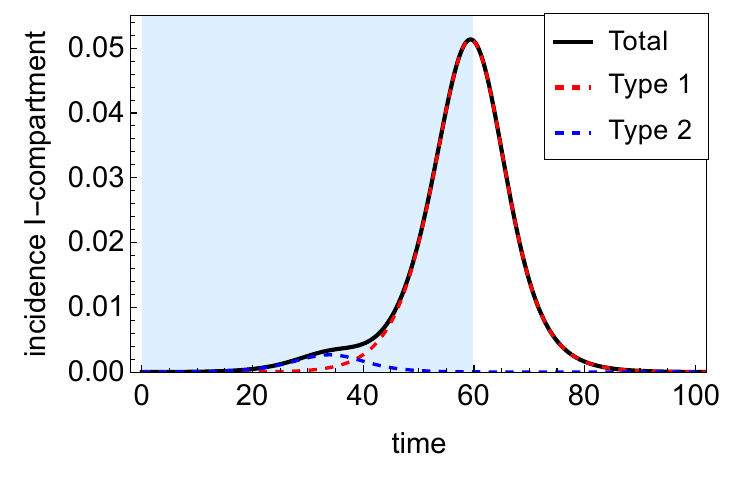}
         \caption{Incidence of infectious cases}  
          \label{incidenceIhitsmall}
     \end{subfigure}
     \hfill
     \begin{subfigure}[b]{0.325\textwidth}
         \centering
         \includegraphics[width=\textwidth]
         {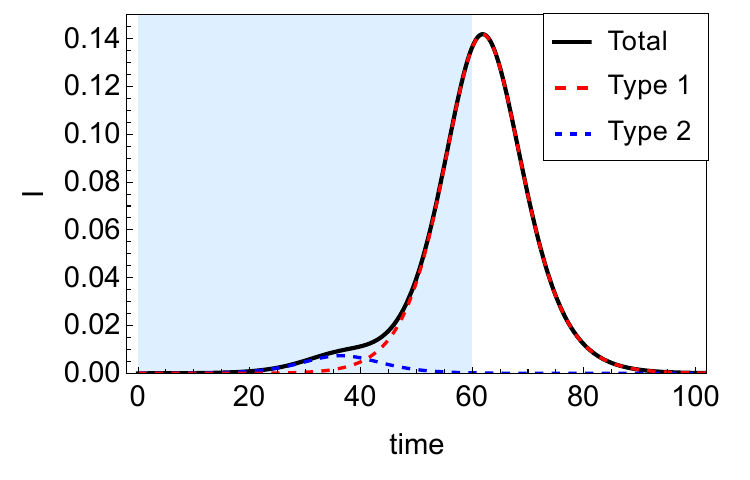}
         \caption{Infectious cases}   
          \label{exheterohitsmallI}
     \end{subfigure}
      \begin{subfigure}[b]{0.325\textwidth}
         \centering
         \includegraphics[width=\textwidth]
         {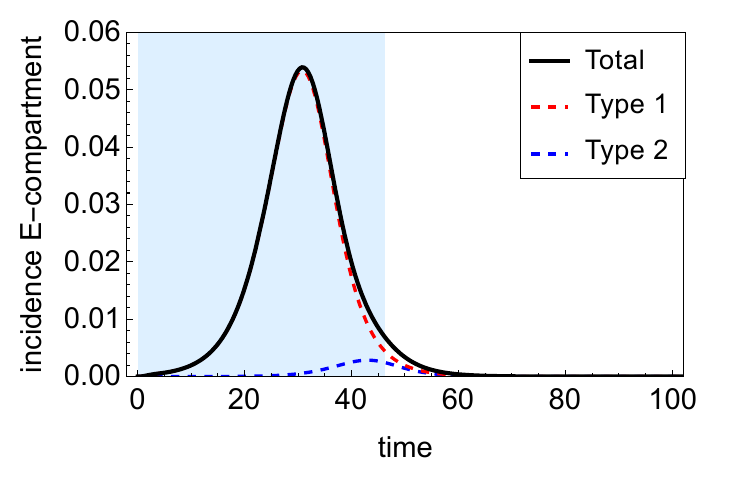}
         \caption{ Incidence of infections}
         \label{incidenceEhitlarge}
     \end{subfigure}
     \hfill
     \begin{subfigure}[b]{0.325\textwidth}
         \centering
         \includegraphics[width=\textwidth]
         {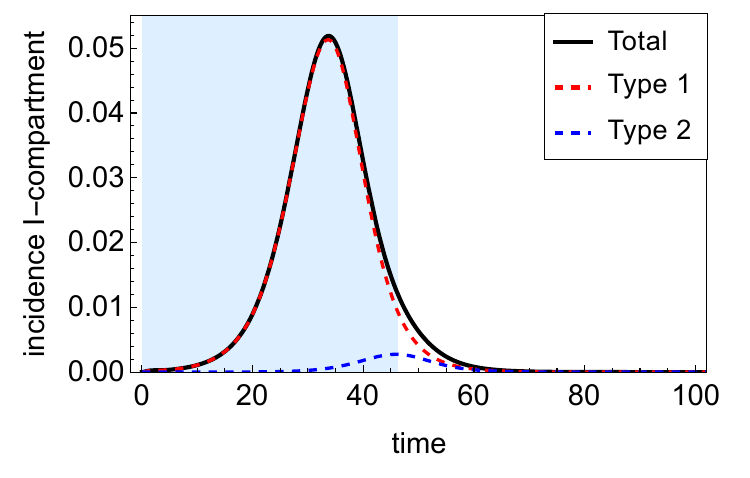}
         \caption{Incidence of infectious cases}  
          \label{incidenceIhitlarge}
     \end{subfigure}
     \hfill
     \begin{subfigure}[b]{0.325\textwidth}
         \centering
         \includegraphics[width=\textwidth]
         {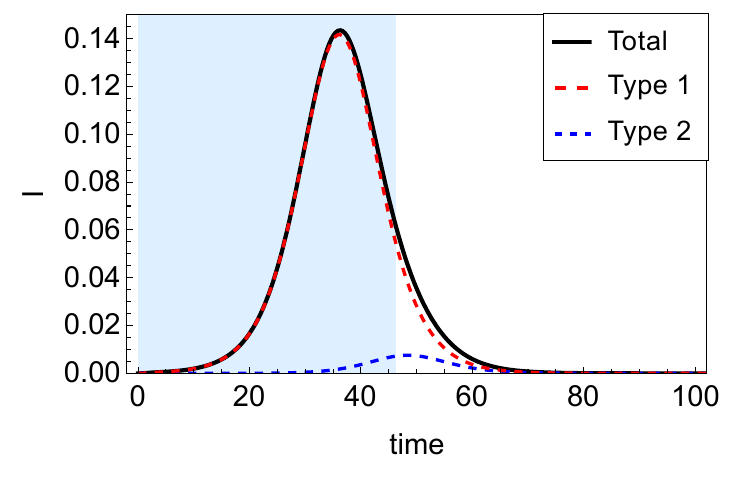}
         \caption{Infectious cases}   
          \label{exheterohitlargeI}
     \end{subfigure}
   \caption{{\bf HIT in case of a small `isolated' community.}
We model two groups, comprising 5\% and 95\% of the population. 99\% of contacts of individuals in the small group are with other individuals in the small group. The groups have the same parameters (the latent period  is Gamma-distributed with shape-parameter 3 and scale parameter  1, the infection period is exponentially distributed with mean 3). In the initial phase, the expected number of secondary cases per primary case is 3, irrespective of the group. Initially there is no immunity. In (a), (b), and (c) a fraction 0.001 of the small population is infectious  while there are initially no infectious individuals in the large population. 
In (d), (e), and (f) a fraction 0.001 of the large population is infectious  while there are initially no infectious individuals in the small population.
The shading changes at the HIT.
 (a) and (d): Incidence of new infections. (b) and (e): Incidence of new infectious individuals. (c) and (f): I-compartment as fraction of the total population}\label{hitstart}
   \end{figure}

\subsection{A simple yet illuminating example (showing, among other things, that there is no clear relationship between $R_0$ and final size)}

   Let $\Omega$ consist of just two points, labeled 1 and 2. Following notational custom, we now represent the last two arguments of $A$ as indices. So $A_{ij}(\tau)$ is the expected force of infection exerted on an individual of type $i$ by an individual of type $j$ that was infected time $\tau$ ago. Similarly we denote $N$ times the fraction of the population that is of type $i$ by $N_i$.

The $2\times 2$ matrix $\mathcal K_0$ is accordingly defined by

\begin{equation}   (\mathcal K_0)_{ij}  =  N_j \int_{0}^{\infty} A_{ij}(\tau) d\tau\end{equation}
 for $i,j\in\{1,2\}$.

As explained in Section 3, $\mathcal K_0$ is the NGM in terms of fractions. The more traditional NGM in terms of numbers is $K$ with
\begin{equation}    K_{ij}  =   N_i \int_0^{\infty} A_{ij}(\tau) d\tau\end{equation}
and the two are related by \eqref{K0_fractions} with $T$ the scaling of the axes given by
\begin{equation}   T x = \left(\begin{array}{c}                  N_1 x_1\\N_2 x_2\\
\end{array}\right)
\end{equation}
   When (i) the type $j$ only affects the contact intensity $c_j$ and (ii) a fraction $c_j N_j / (c_1 N_1 + c_2 N_2)$ of the contacts of any individual is with individuals of type $j$, we are in the separable mixing situation described in the Sections 5 and 6 and

\begin{equation}      A_{ij} (\tau) = \frac{c_i c_j}{ c_1 N_1 + c_2 N_2} b(\tau)\end{equation}

In this case the range of $\mathcal K_0$ is spanned  by      $\left(\begin{array}{c}c_1\\c_2\end{array}\right)$ and this vector therefore is the eigenvector corresponding to the unique non-zero eigenvalue

\begin{equation}     R_0 = \frac{c_1^2 N_1 + c_2^2 N_2}{c_1 N_1 + c_2 N_2}  \int_0^{\infty} b(\tau) d\tau\end{equation}

In terms of the distribution of $c$ in the population, the first factor can be written as

$$      \mbox{mean  +  variance/mean}$$

(this classical result was important in the context of HIV, as it showed that ignoring the variance of the (homo)sexual activity distribution leads to a wrong estimate of $R_0$; the underlying reason is that very active individuals are BOTH more susceptible AND more infectious).
   It may happen that $N_1 \ll N_2$ but $c_1 \gg c_2$ . In such a situation the type 1 subpopulation forms a core group  of  superspreaders, meaning that it is a small group contributing heavily to transmission. Note that the distribution of the two types is:
   
   \begin{equation*}
   \begin{array}{cl}
   N_1 : N_2&\mbox{among all  individuals,}\\
    c_1 N_1 : c_2 N_2 &\mbox{in the incidence during the initial phase of the outbreak,}\\
c_1^2 N_1 : c_2^2 N_2 &\mbox{among the infectors during the initial phase.}
\end{array}
\end{equation*}
   A simple computation shows that

\begin{equation}
    \frac{\partial R_0}{\partial c_2} < 0  \mbox{ if } \frac{N_2}{N_1}  < \frac{c_1}{c_2} \left( \frac{c_1}{c_2} - 2 \right)
\end{equation}

At first it might seem counterintuitive that $R_0$ can DECREASE when individuals of a subgroup INCREASE their contact rate. But once one realises that a side effect might be that the members of the core group have less WITHIN-core-group contacts, it should be clear how the intuition should be up-dated.
   The relation \eqref{s_w} now takes the form

\begin{equation}s_i(t) = e^{- c_i w(t)}\label{si}\end{equation}

while \eqref{wRE} boils down to

\begin{equation}  w(t) = \int_0^{\infty} b(\tau) \left(\frac{c_1 N_1}{c_1 N_1 + c_2 N_2} ( 1 - s_1(t - \tau) )  +  \frac{c_2 N_2}  {c_1 N_1 +c_2 N_2} ( 1 - s_2(t - \tau) ) \right) d\tau\end{equation}

combined with \eqref{si}. In both of these relations we can take the limit $t\to \infty$. This yields an equation for $w(\infty)$ and the identities

\begin{equation}  s_i(\infty) = e^{- c_i w(\infty)}\end{equation}

and from these we can compute the overall final escape fraction

\begin{equation}   s(\infty) := \frac{ s_1(\infty) N_1  +  s_2(\infty) N_2 }{N_1 + N_2}\end{equation}

In the special case $c_1 = c = c_2$ there is, after all, no heterogeneity. In the special case $c_1 = c$ , $c_2 = 0$ the infection only circulates in subpopulation 1. In both special cases we have

\begin{equation}   R_0 = c \int_0^{\infty} b(\tau) d\tau\end{equation}

and the equation

\begin{equation} w(\infty) = \int_0^{\infty} b(\tau) d\tau \left( 1 - e^{-c w(\infty)} \right)\label{winf}\end{equation}

showing that $w(\infty)$ is the same in these two cases.

On the other hand, we have 

\begin{equation}  s(\infty) = 1 - \frac{w(\infty)}{\int_0^{\infty} b(\tau) d\tau}\end{equation}
in the homogeneous case (which allows to rewrite \eqref{winf} in the more familiar form $s(\infty) = e^{-R_0 (1 - s(\infty))}$, while in the `disconnected' case we find

\begin{equation}  s(\infty) = \frac{N_1}{N_1 + N_2} \left(1-\frac{w(\infty)}{\int_0^{\infty} b(\tau) d\tau } \right)+ \frac{N_2}{N_1 + N_2}\end{equation}

clearly showing that for small $N_1$ almost the entire population `escapes' infection, simply since almost all individuals belong to the isolated type 2 subpopulation. By continuity we obtain a similar situation for small values of $c_2$: a small core group of superspreaders has a large impact on $R_0$ but, due to its smallness, not necessarily on the size of the outbreak. We conclude that heterogeneity can `destroy' the monotone relationship between $R_0$ and final size.

\section{Summary and Discussion}
 In 1927, inspired by work of Sir Ronald Ross and Hilda Hudson \cite{Ross1916,Ross1917}, William Ogilvy Kermack and Anderson Gray McKendrick \cite{Kermack1927} introduced a simple yet very powerful idea into the Mathematical Epidemiology of Infectious Diseases: they took as the key model ingredient a description of how a newly infected individual is expected to contribute to the future force of infection on other individuals. More precisely, they employed a function $A$ of a time variable $\tau$, with $A$ the expected contribution to the force of infection and $\tau$ the time elapsed since exposure (often $\tau$ is called ‘infection-age’).
   In exactly this spirit we here incorporate host heterogeneity by taking as the starting point a function $A$ of three variables, $x$, $\tau$ and $\xi$, where $A$ and $\tau$ keep their interpretation, while $x$ specifies the trait of the individual subjected to the force of infection and $\xi$ specifies the trait of the infected individual. Except for a brief excursion in Appendix B, we assume that the trait of an individual is a static characteristic. We use a measure $\Phi$ to describe the trait distribution in the host population, thus unifying models with a subdivision into finitely many discrete traits (then usually called `types') and models incorporating a continuum of traits.
   
   A first aim of our work is to show that many familiar qualitative features of epidemic outbreak models can easily be characterized in this general setting. A highlight, in our opinion, is the formulation of the final size equation in terms of the Next Generation Operator acting on fractions (a precursor of this result can be found in V. Andreasen’s paper \cite{Andreasen2011}, but here we make it more explicit in a much more general setting). Concerning quantitative aspects, our recommendation is to resort to a discrete time formulation as in \cite{Diekmann2021}, but we do not elaborate on this. 
   
   A second aim of our work is to reveal the simplification that occurs when the function $A$ of three variables is actually the product of three functions of one variable, i.e., in case of separable mixing, as it is called in modelling terms. The main result is that in this case we are essentially back to the scalar `homogeneous' case, but with a modified nonlinearity that incorporates the effects of the heterogeneity. Triggered by the Covid-19 outbreak, there has recently been a lot of attention for the influence of heterogeneity on the Herd Immunity Threshold (HIT). The scalar character that results from separable mixing greatly facilitates the characterization of the HIT. When the trait space is $(0,\infty)$ and $\Phi$ is the Gamma Distribution one even obtains explicit expressions showing that the HIT is quickly reached when the variance of the trait distribution is high, a point stressed by G. Gomes and others in \cite{Britton2020,Gomes2022,Montalban2022,Neipel2020,Tkachenko2021}. Here we showed that this holds for general $\tau$ dependence of the infectiousness (and not just for SIR or SEIR compartmental models; please note that, as shown in the companion paper \cite{Diekmann2022}, our general framework allows to derive in a rather easy way separable mixing heterogeneous variants of these and other compartmental models). Not surprisingly, the effect of heterogeneity on the HIT is reduced when the trait is dynamic, see \cite{Tkachenko2021b}
   and Appendix B.
   
   In order to specify the model ingredient $A$ in more detail, one has to reflect on the influence of host population size and composition on contact intensities. This gives rise to a highly nontrivial modelling difficulty: if we have information about the contact structure for a certain host population size and composition, how can we extrapolate this information to situations in which the size and composition are different? A third aim of our work is to call attention to this challenge, partly by reviving the Heesterbeek-Metz pair formation model introduced in \cite{Heesterbeek1993}. Our main contribution to this topic, is to present it as an open problem!

\section*{Acknowledgements}
We thank Horst Thieme and an anonymous referee for helpful suggestions to improve the exposition.

%======== References =========

\appendix
\section{Mathematical statements}

 In this appendix we prove some elementary auxiliary results that are used in the main text. Throughout this section
$z$ denotes a scalar.%$x$ and $y$ denote $m$-vectors, $M$ is an irreducible positive $m\times m$ matrix, $\rho(M)$ is the dominant eigenvalue of the matrix $M$, and for a given $M$, the function $F: \mathbb{R}^m\to \mathbb{R}^m$ is componentwise defined by $F_i(x):=1-e^{-(Mx)_i}$.

%\begin{lemma}
\begin{customlemma}{\!\!A}
\label{lemmaA1}
\hangindent\leftmargini
\textup{} 
\begin{itemize}
\item[i)] $1-e^{-z}\leq z \text{ for }z\geq 0$.
\item[ii)] $1 - e^{-z} \geq \frac{1 - e^{-z_0}}{z_0} z$ for $z_0 > 0$ and $0 \leq z \leq z_0$.
\item[iii)] $\forall \epsilon>0$ $\exists \delta=\delta(\epsilon)$ such that $1-e^{-z}\geq(1-\epsilon)z$ for $0\leq z\leq \delta(\epsilon)$.
\item[iv)] $1-e^{-\theta z}>\theta(1-e^{-z})$ for $z>0$ and $0<\theta<1$.
\item[v)] $(1-z)\log{(1-z)}+\rho z>0$ for $0<z< 1$ and $\rho>1$.
\end{itemize}
%\end{lemma}
\end{customlemma}
\begin{proof}
\hangindent\leftmargini
\textup{} 
\begin{itemize}
\item[i)]  Let $h(z):=1-e^{-z}-z$. Then $h(0)=0$ and $h'(z)=e^{-z}-1<0$ for $z>0$.
\item[ii)] Define $h(z):= 1 - e^{-z} - \frac{1 - e^{-z_0}}{z_0} z$ then $h(0) = 0 = h(z_0)$ and $h'(z) = e^{-z} - \frac{1 - e^{-z_0}}{z_0}$, $h''(z) = -e^{-z} < 0$. We note that $h'(0) = 1 - \frac{1-e^{-z_0}}{z_0} > 0$ (see the proof of Lemma \ref{lemmaA1}.i), while $h'(z_0) = \frac{(z_0 + 1) e^{-z_0} - 1}{z_0} < 0$. It follows that $h(z) > 0$ for $0 < z < z_0$.
\item[iii)] Let now $h(z):=1-e^{-z}-(1-\epsilon)z$. Then $h(0)=0$ and $h'(z)=e^{-z}-(1-\epsilon)$. So $h'(0)=\epsilon>0$ and consequently $h$ is positive for small positive $z$.
\item[iv)] Define, for given $z>0$, $H(\theta):=1-e^{-\theta z}-\theta(1-e^{-z})$. Then $H$ is continuous, $H(0)=0$, $H(1)=0$, $H'(\theta)=ze^{-\theta z}-1+e^{-z}$
and $H''(\theta)=-z^2e^{-\theta z}<0$. So $H$ cannot have two or more zero's in between 0 and 1. Since $H'(0)=z-1+e^{-z}>0$ (see Lemma \ref{lemmaA1}{.i)}), $H$ is positive for $0<\theta<1$.
\item[v)] Let $h(z):=(1-z)\log{(1-z)}+\rho z$, then $h(0)=0$ and $h'(z)=-1-\log{(1-z)}+\rho>0$ for $0\leq z<1$, so $h$ is positive on $[0,1]$.
\end{itemize}
\end{proof}

\section{Dynamic heterogeneity}\label{appendix_dynamic}

So far we considered a host population consisting of individuals with different STATIC traits, where `static' refers to the assumption that individuals do not change trait as the pathogen spreads through the population. The aim of this appendix is to warn readers that results may change significantly when the trait itself is a dynamic variable. 
%the differences in the mechanism underlying static and dynamic heterogeneity may have material impact on the outcome of an epidemic.
%results under static heterogeneity do not necessarily transfer to the situation with dynamic heterogeneity.

We distinguish:

\begin{enumerate}
\item [1] models with trait-dynamics \textit{not} influenced by disease status,
\item [2] models with trait-dynamics influenced by disease status,
\end{enumerate}

with the second class of models typically requiring more model assumptions and exhibiting more complex dynamics.

%We propose a class 1 model in the next subsection to illustrate how a model with dynamic heterogeneity can be studied in relation to its static variant. Furthermore, we will show numerically how results in a static setting changes when a heterogeneous setting is considered instead.

%For models with trait-dynamics influenced by disease status we discuss models as considered in [x,y,z] and illustrate the abundance of ways to incorporate disease status in trait-dynamics. With this note we like to emphasize the importance of finding heuristic guidelines to this class of models. To our best knowledge, this remains an outstanding challenge.

\subsection{dynamic heterogeneity: \textit{not} influenced by disease status} We introduce a model with the goal to study the impact of dynamic heterogeneity in a framework where dynamics are not influenced by disease status. Individuals are characterized by the number of contacts they make per unit of time. Type $1$ has contact rate $c_1$, type $2$ has contact rate $c_2$. Transitions occur at rates
\begin{equation}
\nu \begin{pmatrix}-\theta & \frac{\theta}{\theta-1} \cr \theta & -\frac{\theta}{\theta-1}\cr \end{pmatrix},
\end{equation}
with $\nu>0$ the frequency (note that a round trip takes on average $\frac{1}{\nu\theta}+\frac{\theta-1}{\nu \theta}=\frac{1}{\nu}$) and $\theta>1$ a measure for the asymmetry in the sojourn times, and hence for the asymmetry of the two steady state subpopulations.
We assume the subpopulations to be in their steady states. Hence, if $N$ denotes the total population size and $N_1$ and $N_2$ the subpopulation sizes then 
\begin{equation}\label{1.2}
\begin{pmatrix} N_1\cr N_2 \end{pmatrix}=\frac{N}{\theta}\begin{pmatrix} 1\cr \theta-1 \end{pmatrix}.
\end{equation}

The average contact rate $c$ is hence given by
\begin{equation} \label{contact_rate}
c=\frac{c_1}{\theta}+\frac{\theta-1}{\theta}c_2=c_2+\frac{c_1-c_2}{\theta}.
\end{equation}
As a normalization, we fix $c$ and require that $c_2 \le c$.

It follows that there are three ``free'' parameters: $\nu>0$, $\theta>1$, $0 \le c_2 \le c$. We assume proportionate mixing: if an individual makes a contact, it is with probability
\begin{equation}
\frac{c_1N_1}{c_1N_1+c_2N_2}
\end{equation}
with an individual of type $1$.

\subsubsection{The SIR model}

We will combine the `dynamic heterogeneous' ingredients in the previous section with a homogeneous version of the epidemic model, the SIR model, as described by
\begin{equation}
\begin{aligned}
&\frac{dS}{dt}=-\beta c \frac{SI}{N} \cr
&\frac{dI}{dt}=\beta c \frac{SI}{N}-\alpha I, \cr
\end{aligned}
\end{equation}
with $S$ and $I$ the amount of susceptible and infectious individuals respectively. Here $\beta$ with $0<\beta \le1$, is the probability of transmission in a contact between an infectious and a susceptible individual.  The expected duration of the infectious period equals $\frac{1}{\alpha}$.
By scaling of the time variable, we achieve that $\alpha=1$.
By scaling of $c$ we achieve that $\beta=1$.
We now list some relevant features:

\begin{align}
{\rm basic~reproduction~number~} R_0 &=c, \label{R_0_homog}\\ 
{\rm herd~immunity~threshold~(HIT)} ~ \bar{s}&=\frac{\bar{S}}{N}=\frac{1}{R_0}=\frac{1}{c}, \label{HIT_SIR}\\
{\rm final~size~equation}~~s(\infty)&=e^{-c(1-s(\infty))}.\label{finalsize_SIR}
\end{align}

We define the overshoot as $\bar{s}-s(\infty)$.

\subsubsection{The combined model}

Recall that we assume that type transitions are not influenced by disease status and that $\alpha=1$ and $\beta=1$.
Recall that
\begin{equation}
c=c_1\frac{N_1}{N}+c_2\frac{N_2}{N}=\frac{c_1}{\theta}+(1-\frac{1}{\theta})c_2.
\end{equation}
Keep in mind that if a contact is with a type $i$ individual, it is with probability $\frac{S_i}{N_i}$ with a susceptible individual. The combined model is defined as follows

\begin{equation}
\begin{aligned}
&\frac{dS_1}{dt}=-\frac{c_1S_1}{c_1N_1+c_2N_2}(c_1I_1+c_2I_2)-\nu \theta S_1+\nu \frac{\theta}{\theta-1}S_2 \cr
&\frac{dS_2}{dt}=-\frac{c_2S_2}{c_1N_1+c_2N_2}(c_1I_1+c_2I_2)+\nu \theta S_1-\nu \frac{\theta}{\theta-1}S_2 \cr
&\frac{dI_1}{dt}=\frac{c_1S_1}{c_1N_1+c_2N_2}(c_1I_1+c_2I_2)-\nu \theta I_1+\nu \frac{\theta}{\theta-1}I_2-I_1 \cr
&\frac{dI_2}{dt}=\frac{c_2S_2}{c_1N_1+c_2N_2}(c_1I_1+c_2I_2)+\nu \theta I_1-\nu \frac{\theta}{\theta-1}I_2-I_2 \cr
\end{aligned}
\end{equation}

Define $s_j=\frac{S_j}{N}$ and $i_j=\frac{I_j}{N}$.  Then

\begin{equation} \label{combined_model_ode}
\begin{aligned}
&\frac{ds_1}{dt}=-\frac{c_1}{c}(c_1i_1+c_2i_2)s_1-\nu \theta s_1+\nu \frac{\theta}{\theta-1}s_2 \cr
&\frac{ds_2}{dt}=-\frac{c_2}{c}(c_1i_1+c_2i_2)s_2+\nu \theta s_1-\nu \frac{\theta}{\theta-1}s_2 \cr
&\frac{di_1}{dt}=\frac{c_1}{c}(c_1i_1+c_2i_2)s_1-\nu \theta i_1+\nu \frac{\theta}{\theta-1}i_2-i_1 \cr
&\frac{di_2}{dt}=\frac{c_2}{c}(c_1i_1+c_2i_2)s_2+\nu \theta i_1-\nu \frac{\theta}{\theta-1}i_2-i_2 \cr
\end{aligned}
\end{equation}

\subsubsection{Static heterogeneity}\label{static_hetero}

In the combined model of the previous subsection we put $\nu=0$, but keep \eqref{1.2}. The resulting model is one with static heterogeneity. If we freeze the values $s_1$ and $s_2$ and introduce $x:=c_1i_1+c_2i_2$ (as a metric for the subpopulation of infectious individuals), we obtain
\begin{equation}
\frac{dx}{dt}=\left(\frac{c_1^2}{c}s_1+\frac{c_2^2}{c}s_2-1\right)x.
\end{equation}
We conclude that
\begin{equation}\label{static_R0}
R_0=\frac{c_1^2}{c}\frac{1}{\theta}+\frac{c_2^2}{c}(1-\frac{1}{\theta})
\end{equation}
and that the HIT is characterized by
\begin{equation}
\frac{c_1^2}{c}\bar{s}_1+\frac{c_2^2}{c}\bar{s}_2=1.
\end{equation}

The latter can be reduced to an equation for the scalar variable $w$, where
\begin{equation}
s_1(t)=\frac{1}{\theta}e^{-c_1 w(t)},  ~~s_2(t)=(1-\frac{1}{\theta})e^{-c_2 w(t)}. \label{expr_s1_s2}
\end{equation}
The equation for the HIT in $w$ reads 
\begin{equation} \label{HIT_eq_w}
\frac{c_1^2}{c}\frac{1}{\theta}e^{-c_1\bar{w}}+\frac{c_2^2}{c}(1-\frac{1}{\theta})e^{-c_2\bar{w}}=1,
\end{equation}
and for that value of $w$ the HIT, as a fraction of the total population, is given by
\begin{equation}
{\rm HIT}: ~ ~\bar{s}_1+\bar{s}_2=\frac{1}{\theta}e^{-c_1\bar{w}}+(1-\frac{1}{\theta})e^{-c_2\bar{w}}.
\end{equation}

The equation
\begin{equation}
\frac{dw}{dt}=-w+\frac{c_1}{c}\frac{1}{\theta}(1-e^{-c_1 w})+\frac{c_2}{c}(1-\frac{1}{\theta})(1-e^{-c_2 w}),
\end{equation}
can be easily derived by combining the defining relation (see \eqref{combined_model_ode})
\begin{equation}
\frac{dw}{dt}=\frac{c_1}{c}i_1+\frac{c_2}{c}i_2, \label{force_inf_eq}
\end{equation}
with
\begin{equation}
\frac{ds_j}{dt}+\frac{di_j}{dt}=-i_j, \quad j=1,2 \label{net_out}
\end{equation}
(Use \eqref{net_out} to write the r.h.s. of \eqref{force_inf_eq} as a time derivative; next integrate from $-\infty$ to $t$ and use \eqref{expr_s1_s2} and \eqref{force_inf_eq}).

The limit $w(\infty)$ is accordingly characterized by
\begin{equation}\label{finalsize_eq_w}
w(\infty)=\frac{c_1}{c}\frac{1}{\theta}(1-e^{-c_1 w(\infty)})+\frac{c_2}{c}(1-\frac{1}{\theta})(1-e^{-c_2 w(\infty)}).
\end{equation}
The fraction that escapes infection is given by
\begin{equation}
s_1(\infty)+s_2(\infty)=\frac{1}{\theta}e^{-c_1w(\infty)}+(1-\frac{1}{\theta})e^{-c_2w(\infty)},
\end{equation}
and the overshoot by
\begin{equation}
\bar{s}_1+\bar{s}_2-s_1(\infty)-s_2(\infty).
\end{equation}

We obtain by combining \eqref{contact_rate} with \eqref{static_R0}
\begin{equation}
    R_0 = \frac{\theta}{c}\Big((1 - \frac{1}{\theta})c_2^2 -2 c (1 - \frac{1}{\theta}) c_2 + c^2  \Big).
\end{equation}

The first and second derivative of $R_0$ to $c_2$ are
\begin{equation}
    \frac{dR_0}{dc_2} = \frac{\theta}{c}\Big(2(1 - \frac{1}{\theta})c_2 -2 c (1 - \frac{1}{\theta}) \Big)
\end{equation}
and
\begin{equation}
    \frac{d^2R_0}{dc_2^2} = \frac{\theta}{c}\Big(2(1 - \frac{1}{\theta}) \Big)
\end{equation}
respectively. Recall $\theta > 1$, hence $R_0$ is a quadratic function in $c_2$ with strictly positive second derivative and minimum in $c_2 = c$ with value $c$. Thus, when $c_2 \neq c$ we have $R_0 > c$ for a model with static heterogeneity ($\nu = 0$).

\subsubsection{Comparing $R_0$, $s(\infty)$ and HIT in the extremes, i.e., $\nu=0$ and $\nu=\infty$}

We consider the homogeneous SIR model as describing the limit $\nu \to \infty$ (we refrain from providing a formal justification). We then have $R_0 = c$, see \eqref{R_0_homog}. Using the statement in the last paragraph of Section \ref{static_hetero} we can conclude that $R_0(\nu = 0)> R_0(\nu = \infty)$ for $c_2 \ne c$. Thus, when the contact rate of different types differ, the basic reproduction number is strictly higher when considered in a ($\nu = 0$) heterogeneous setting than in a ($\nu = \infty$) homogeneous setting. 

For the final size $s(\infty)$ and HIT we find through numerical investigation a similar relation between $\nu = 0$ and $\nu = \infty$ as shown in Figure \ref{fig:HIT_final_size}. 

\begin{figure}[htbp] 
\centering
\includegraphics[width=0.8\linewidth]{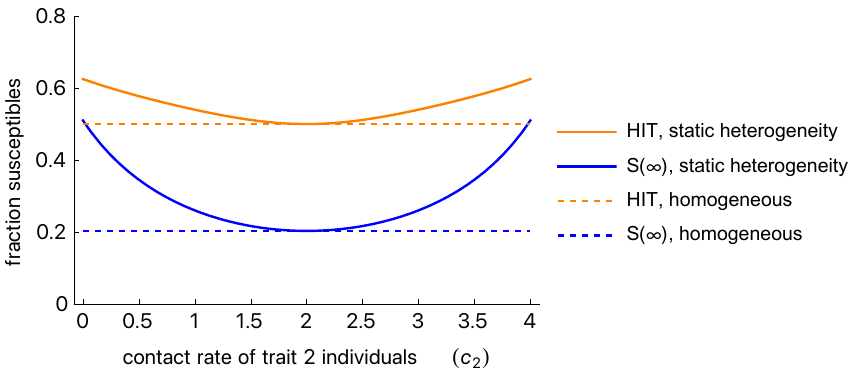} 
\caption{HIT and the final size as a function of contact rate $c_2$ in a static heterogeneous setting ($\nu=0$) and in a homogeneous setting ($\nu=\infty$). $\theta=2$. $c_1$ is a function of $c_2$ such that the average contact rate equals $c=2$. The results for $\nu = 0$ follow from numerically solving equations \eqref{HIT_eq_w} and \eqref{finalsize_eq_w}. The results for $\nu = \infty$ follow directly from \eqref{HIT_SIR} and from solving \eqref{finalsize_SIR}.}\label{fig:HIT_final_size}
\end{figure}

\subsubsection{Numerical investigation of dynamic heterogeneity ($0 < \nu < \infty$)}
By interpolating the results in the extremes one may tend to conclude that $R_0$, the HIT and final size are a decreasing function of the rate of trait change $\nu$. However, numerical investigation for $0 < \nu < \infty$ shows otherwise, see Figure \ref{fig:nu_final_size}.

\begin{figure}[htbp] 
\centering
\includegraphics[width=0.5\linewidth]{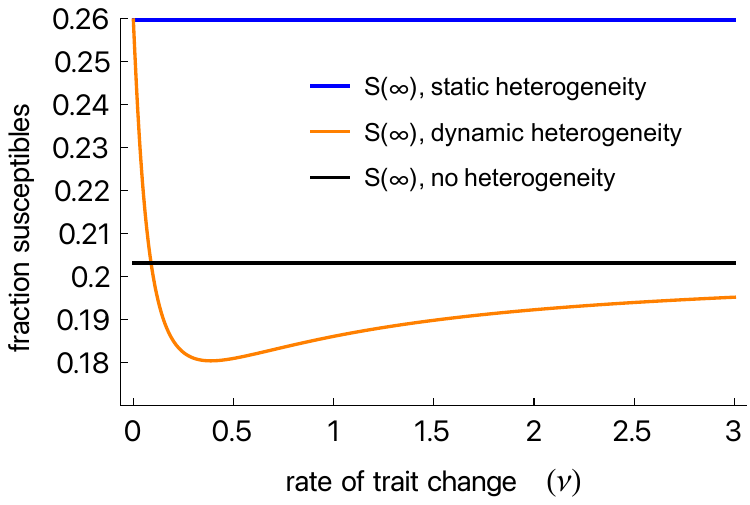} 
\caption{Final size as a function of $\nu$ in a homogeneous, a static heterogeneous and a dynamic heterogeneous setting. $\theta=2$, $c = 2$ , $c_1 = 3$, $c_2 = 1$ The results follow from numerically solving \eqref{combined_model_ode}.}\label{fig:nu_final_size}
\end{figure}

Our interpretation of Figure \ref{fig:nu_final_size} suggests 3 types of mechanisms at play when considering dynamic heterogeneity (not influenced by disease-status).

\begin{itemize}
    \item[1] When heterogeneity is (nearly) static, or $\nu=0$, individuals of the group with higher contact rate infect more individuals, but also are being infected more. One would therefore expect a relatively fast and large outbreak in the group with higher contact rate. After this, the dissemination is mainly driven by individuals with lower contact rate. This causes relatively less infections among the lower contact rate individuals and eventually less infections in total.
    \item[2] A new mechanism emerges in addition to (1) when heterogeneity is dynamic and $\nu$ is of similar order as the rate of susceptible individuals becoming infected. While mechanism (1) still partly holds, the number of higher-contact-rate susceptible individuals decreases  slowly due to `inflow' of susceptible lower-contact-rate individuals. This leads to more infections overall.
    \item[3] When the rate of trait-change is much higher than the rate of susceptibles becoming infected, we find the impact of mechanisms (1) and (2) being diluted. This does not come as a surprise as individuals can hardly be distinguished anymore from individuals with an average contact rate $c$ in a homogeneous model.
\end{itemize}

Our main message is that new mechanisms arise when a model is considered incorporating a dynamic heterogeneity setting instead of a static heterogeneity one. Even when the total number of infections is lower in a static heterogeneity setting than in a homogeneous one, the result can be opposite when considering from a dynamic heterogeneity setting as shown in Figure \ref{fig:nu_final_size}.

\subsection{Models in which trait dynamics incorporates feedback (in particular information about incidence, prevalence and/or own disease status)}

The recent outbreak of Covid-19 provides much motivation for formulating and analysing models that incorporate dynamic heterogeneity and allow the dynamics of the trait of an individual to be influenced by (information about) both the population level epidemic dynamics and its own health (and vaccination) status. We refer to \cite{Juher2023,Tang2022,Teslya2022,Tkachenko2021, Berestycki2023, Zhang2023} for recent examples, to \cite{Wambua2023} for data aspects and to \cite{Liu2007, Poletti2009,Funk2010, Just2018, d'Onofrio2009, Cui2008} for pre-Covid pioneering work. The subject is still in its infancy and we expect much more work in years to come.

To organize our thoughts, we first focus on trait dynamics, while assuming that, both at the individual level and at the population level, the disease related dynamics is known/given. This allows us to think of the trait as a stochastic variable that follows a Markov process in a non-constant environment, so with time-dependent transition rates. The modelling task is to specify these transition rates. Here we limit ourselves to the simpler task of indicating the variables on which the transition rates depend (without specifying the dependence itself).

The following classification attempts to structure the sea of possibilities a little bit, without claiming completeness. These rates may depend on:
\begin{enumerate}
\item the health status of the individual itself
\item the perceived current incidence or prevalence (with the perception being based on information; how information is handled, in particular whether it is trusted, may depend in part on the trait itself); in network models one can work with a local version based on information about the health status of acquaintances
\item governmental advice/rules (with compliance possibly depending on the trait itself)
\end{enumerate}
Likewise one can describe the epidemic dynamics while pretending that the dynamic trait distribution is known. By coupling the two time-inhomogeneous dynamical systems one then obtains, through a fixed point argument, an autonomous nonlinear dynamical system. If time scale differences exist (for instance, trait dynamics might be fast relative to the epidemic time scale) these should be exploited to facilitate the analysis!

Admittedly the above is a rather sketchy description of a huge class of models. It is intended as a stimulus, as an invitation, and not as a solid exposition.


\begin{thebibliography}{99}

\bibitem{Almeida2021}
L. Almeida, P. Bliman, G. Nadin, B. Perthame and N. Vauchelet,  {\textit{Final size and convergence rate for an epidemic in heterogeneous populations}}, 
Math. Models Methods Appl. Sci. {\bf 31}(2021), 1021-1051.

\bibitem{Andreasen2011}
V. Andreasen, {\textit{The Final Size of an Epidemic and Its Relation to the Basic Reproduction Number}}, Bull. Math. Biol. {\bf 73}(2011), 2305-2321. 

\bibitem{Avram2023}
F. Avram, R. Adenane, L. Basnarkov, G. Bianchin, D. Goreac and A. Halanay, {\textit {An Age of Infection Kernel, an R Formula, and Further Results for Arino-Brauer A, B Matrix Epidemic Models with Varying Populations, Waning Immunity, and Disease and Vaccination Fatalities.}}, Mathematics {\bf 11}(2023):1307.

\bibitem{Ball2023}
F. Ball, L. Critcher, P. Neal and D. Sirl, {\textit{The impact of household structure on disease-induced herd immunity.}}  J. Math. Biol. {\bf 87}, 83 (2023), 83. 

\bibitem{Barril2023}
C. Barril, P. Bliman and S. Cuadrado, {\textit{Final Size for Epidemic Models with Asymptomatic Transmission}},  Bull. Math. Biol. {\bf 85}(2023), 52.

\bibitem{Batkai2017}
A. B{\'a}tkai, M. Kramar Fijav{\v{z}} and A. Rhandi, {\it {Positive Operator Semigroups : from Finite to Infinite Dimensions}}, Birkha{\"u}ser, Basel, 2017.

\bibitem{Berestycki2023}
H. Berestycki, B. Desjardins, J. Weitz and J. Oury, {\it {Epidemic modeling with heterogeneity and social diffusion}}, J. Math. Biol. {\bf 86}(2023), 60.

\bibitem{Bootsma2022}
M. Bootsma, B. Kolen, 
M. de Vries, M. Kretzschmar,
and N. Mouter. {\it{Effectiviteit van verschillende toepassingen van het Coronatoegangsbewijs }} (in Dutch)(2022) https://open.overheid.nl/repository/ronl-aa2988ce-324f-4c1c-ad39-459358e32bfe/1/pdf/effectiviteit-coronatoegangsbewijs-eindversie-tu-delft4.pdf

\bibitem{Bootsma2023}
M. Bootsma, K. Chan, O. Diekmann and H. Inaba. {\it {Separable mixing: The general formulation and a particular example focusing on mask efficiency}}, Math. Biosci. Eng. {\bf 20}(2023), 10.

\bibitem{Brauer2009}
F. Brauer and J. Watmough, {\it {Age of infection epidemic models with heterogeneous mixing}}, J. Biol. Dyn., {\bf 3}(2009), 324-330.

\bibitem{Britton2020}
T. Britton, F. Ball and P. Trapman,  {\textit{A mathematical model reveals the influence of population heterogeneity on herd immunity to SARS-CoV2}}, Science {\bf 369}(2020), 846-849. 

\bibitem{Breda2012}
D Breda, O Diekmann, W. de Graaf, A Pugliese and R Vermiglio, {\textit {On the formulation of epidemic models (an appraisal of Kermack and McKendrick)}}, J. Biol. Dyn. {\bf 6}(2012), 103-117. 

\bibitem{Breda2021}
D. Breda, F. Florian, J. Ripoll and R. Vermiglio,  {\textit{ Efficient numerical computation of the basic reproduction number for structured populations}}, J. Comp. Appl. Math. {\bf 384}(2021), 113165.

\bibitem{Breda2020}
D. Breda, T. Kuniya, J. Ripoll and R. Vermiglio,  {\textit{ Collocation of next-generation operators for computing the basic reproduction number of structured populations}}, J. Sci. Comput. {\bf 85}(2020), 40.

\bibitem{Breda}
D. Breda, S. De Reggi, F. Scarabel, R. Vermiglio and J. Wu,  {\textit {Bivariate collocation for computing $R_0$ in epidemic models with two structures}}, submitted.


\bibitem{Castro2020}
M. Castro, S. Ares, J. A. Cuesta and S. Manrubia, {\textit{ The turning point and end of an expanding epidemic cannot be precisely forecast}},  Proc. Natl. Acad. Sci. {\bf 117}(2020), 26190-26196.

\bibitem{Chan2013}
K. Chan, {\textit {Impact of demographical change on the severity of an epidemic}}, Master thesis, Utrecht Univeristy, 2013, https://studenttheses.uu.nl/handle/20.500.12932/13069.

\bibitem{Chan2016}
K. Chan, H. Nishiura, O Diekmann, M. Bootsma, {\textit {The impact of population ageing on the severity of an epidemic outbreak}},  unpublished manuscript, (2016).

\bibitem{Cui2008}
J. Cui, Y. Sun and Huaiping Zhu, {\it {The impact of media on the control of infectious diseases}},
J. Dyn. Diff. Equa. {\bf 20}(2008), 31-53.


\bibitem{Diekmann2008}
O. Diekmann, Ph. Getto and M. Gyllenberg,  {\textit {Stability and bifurcation analysis of Volterra functional equations in the light of suns and stars}}, SIAM J. Math. Anal. {\bf 39}(2008), 1023-1069.

\bibitem{Diekmann2007}
O. Diekmann and M. Gyllenberg,  {\textit {Abstract delay equations inspired by population dynamics}}, In : {\it Functional Analysis and Evolution Equations. The G{\"u}nter Lumer Volume}. 
Birkhauser, Basel, 2007, 187-200.

\bibitem{Diekmann2011}
O. Diekmann and M. Gyllenberg,  {\textit {Equations with infinite delay: blending the abstract and the concrete}},  J. Differ. Equ. {\bf 252}(2011), 819-851.



\bibitem{Diekmann2013}
O. Diekmann, J. Heesterbeek and T. Britton, {\it{ Mathematical Tools for Understanding Infectious Disease Dynamics}}, Princeton University Press, Princeton, 2013.

\bibitem{Diekmann2018}
O. Diekmann, M. Gyllenberg and J. Metz,  {\it {Finite Dimensional State Representation of Linear and Nonlinear Delay Systems}} J. Dyn. Diff. Equat. {\bf 30}(2018), 1439-1467. 

\bibitem{Diekmann2021}
O. Diekmann, H. Othmer, R. Planqu{\'e} and M. Bootsma, {\it {The discrete-time Kermack-McKendrick model: A versatile and computationally attractive framework for modeling epidemics}}, Proc. Natl. Acad. Sci. {\bf 118}(2021),  e2106332118.

\bibitem{Diekmann2022}
O. Diekmann and H. Inaba, {\it {A systematic procedure for incorporating separable static heterogeneity into compartmental epidemic models}}, J. Math. Biol. {\bf 86}(2023), 29.

\bibitem{d'Onofrio2009}
A. d'Onofrio and P. Manfredi, {\it {Information-related changes in contact patterns may trigger oscillations in the endemic prevalence of infectious diseases}}, J. Theor. Biol. {\bf 256}(2009), 473-478.

\bibitem{Eikenberry2020}
S. Eikenberry, M. Mancuso, E. Iboi, T. Phan, K. Eikenberry, Y. Kuang, E. Kostelich and A. Gumelet, {\it {To mask or not to mask: Modeling the potential for face mask use by the general public to curtail the COVID-19 pandemic}}, Infect. Dis. Model. {\bf 5}2020, 293-308. 

\bibitem{Franco2022}
E. Franco, O. Diekmann and M. Gyllenberg, {\it {Modelling physiologically structured populations: renewal equations and partial differential equations}}, J. Evol. Equ. {\bf 23}(2023), 46.

\bibitem{Funk2010}
S. Funk, M. Salath{\'e} and V. Jansen, {\textit {Modelling the influence of human behaviour on the spread of infectious diseases: a review}},  
J. R. Soc. Interface {\bf 7}(2010), 1247-1256.


\bibitem{Gomes2022}
M. Gomes, M. Ferreira, R. Corder, J. King, C. Souto-Maior, C. Penha-Gon{\c{c}}alves,  G. Gon{\c{c}}alves,  M. Chikina, W. Pegden and R Aguas,
 {\it {Individual variation in susceptibility or exposure to SARS-CoV-2 lowers the herd immunity threshold}}, J. Theor. Biol. {\bf 540}(2022), 111063.
 
\bibitem{Hadeler2012}
K. Hadeler, {\it {Pair formation}}, J. Math. Biol. {\bf 64}(2012), 613-645.

\bibitem{Heesterbeek1993}
J. Heesterbeek and J. Metz, {\it {The saturating contact rate in marriage and epidemic models}}, J. Math. Biol. {\bf 31}(1993), 529-539.  

\bibitem{Hill2023}
A. Hill, J. Glasser and Z. Feng, {\it {Implications for infectious disease models of heterogeneous mixing on control thresholds}}, J. Math. Biol. {\bf 86}(2023), 53.


\bibitem{Inaba2012}
H. Inaba,  {\it{On a new perspective of the basic reproduction number in heterogeneous environments}}, J. Math. Biol. {\bf 65}(2012), 309-348.


\bibitem{Inaba2014}
H. Inaba,  {\it{On a pandemic threshold theorem of
the early Kermack-McKendrick model
with individual heterogeneity}}, Math. Popul. Stud. {\bf 21}(2014), 95-111.


\bibitem{Inaba2017}
H. Inaba, {\it {Age-Structured Population Dynamics in Demography and Epidemiology}}, Springer, Singapore, 2017.

\bibitem{Inaba2023}
H. Inaba, {\it {Basic concepts for the Kermack and McKendrick model with static heterogeneity}}, arXiv, preprint, 2023, https://arxiv.org/abs/2311.11247.

\bibitem{Juher2023}
D. Juher, D. Rojas and J. Salda{\~n}a, {\textit {Saddle-node bifurcation of limit cycles in an epidemic model with two levels of awareness}}, Phys. D: Nonlinear Phenom. {\bf 448}(2023), 133714.

\bibitem{Just2018}
W. Just, J. Salda{\~n}a and Y. Xin, {\textit {Oscillations in epidemic models with spread of awareness}}, J. Math. Biol. {\bf 76}(2018), 1027-1057.


\bibitem{Katriel2012}
G. Katriel, {\it {The size of epidemics in populations with heterogeneous susceptibility}}, J. Math. Biol. {\bf 65}(2012), 237-262.

\bibitem{Kermack1927}
W. Kermack and A. McKendrick, {\it {A contribution to the mathematical theory of epidemics}}, Proc. R. Soc. Lond. A {\bf 115}(1927), 700-721.

\bibitem{Li2020}
X. Li, J. Yang and M. Martcheva, {\it {Age Structured Epidemic Modeling}}, Springer, Cham, (2020).

\bibitem{Liu2007}
R. Liu, J. Wu and H. Hu, {\it {Media/psychological impact on multiple outbreaks of emerging infectious diseases}},
Comput. Math. Methods Med. {\bf 8}(2007), 153-164.


\bibitem{Manrubia2020}
S. Manrubia, {\it {The Uncertain Future in How a
Virus Spreads}}, Physics {\bf 13}(2020),166. 


\bibitem{Messina2022a}
E. Messina, M. Pezzella and A. Vecchio, {\it{Positive numerical approximation of an integro-differential epidemic model}}, Axioms {\bf 11}(2022), 69.

\bibitem{Messina2022b}
E. Messina, M. Pezzella and A. Vecchio, {\it {A non-standard numerical scheme for an age-of-infection epidemic model}}, J. Comput. Dyn. {\bf 9}(2022),
239-252.

\bibitem{Messina2023_a}
E. Messina, M. Pezzella and A. Vecchio, {\it {A long-time behavior preserving numerical scheme for age-of-infection epidemic models with heterogeneous mixing}},
Appl. Numer. Math. (2023).

\bibitem{Messina2023}
E. Messina, M. Pezzella and A. Vecchio, {\it{Asymptotic solutions of non-linear implicit Volterra discrete
equations}}, J. Comput. Appl. Math. {\bf 425}(2023), 115068.

\bibitem{Messina2023b}
E. Messina, M. Pezzella and A. Vecchio. {\it {Nonlocal finite difference discretization of a class of renewal equation models for epidemics}}, Math. Biosci. Eng. {\bf 20}(2023), 11656-11675.

\bibitem{Montalban2022}
A. Montalb{\'a}n, R. Corder and M. Gomes, {\it {Herd immunity under individual variation and reinfection}}, J. Math. Biol. {\bf 85}(2022), 2.

\bibitem{Mossong2008}
J. Mossong, N. Hens,
    M. Jit,
    Ph. Beutels,
    K. Auranen,
    R. Mikolajczyk,
    M.  Massari,
    S. Salmaso,
    G. Scalia Tomba,
    J. Wallinga,
    J. Heijne,
    M. Sadkowska-Todys,
    M. Rosinska and
    W. Edmunds
{\it {Social contacts and mixing patterns relevant to the spread of infectious diseases}}, PLOS Medicine {\bf 5}(2008), e74.

\bibitem{Neipel2020}
J. Neipel, J. Bauermann, S. Bo, T. Harmon and F. J\"ulicher,  {\it {Power-Law Population Heterogeneity Governs Epidemic Waves}}, PLoS ONE {\bf 15}(2020), e0239678.

\bibitem{Ngonghala2020}
C. Ngonghala, E. Iboi, S. Eikenberry, M. Scotch, C. MacIntyre, M. Bonds and A. Gumel, {\it {Mathematical assessment of the impact of non-pharmaceutical interventions on curtailing the 2019 novel
Coronavirus}}, Math. Biosci. {\bf 325}(2020).

\bibitem{Nguyen2023}
M. Nguyen, A. Freedman, S. Ozbay and S. Levin, {\it {Power-Law Population Heterogeneity Governs Epidemic Waves}}, submitted.


\bibitem{Novozhilov2008}
A. Novozhilov, {\it{ On the spread of epidemics in a closed heterogeneous population}} Math. Biosci. {\bf 215}(2008), 177-185.

\bibitem{Oliva2023}
G. Oliva, S. Bonfigli, P. Cavallo and A. Scala, {\textit {Navigating the Herd Immunity Surface: A Novel Framework for Optimising Epidemic Response Strategies}}, Qeios (2023).

\bibitem{Pastor-Satorras2022}
R. Pastor-Satorras and C. Castellano, {\it {The advantage of self-protecting interventions in mitigating epidemic circulation at the community level}}, Sci. Rep. {\bf 12}(2022).

\bibitem{Poletti2009}
P. Poletti, B. Caprile, M. Ajelli, A. Pugliese and S. Merler, {\textit {Spontaneous behavioural changes in response to epidemics}}, J. Theor. Biol. {\bf 260}(2009), 31-40.

\bibitem{Rass2003}
L. Rass and J. Radcliffe,  {\it {Spatial Deterministic Epidemics}}, American Mathematical Society, Providence, RI, 2003.

\bibitem{Scarabel2021}
F. Scarabel, O. Diekmann and R. Vermiglio, {\it{
Numerical bifurcation analysis of renewal equations via pseudospectral approximation}},
J. Comput. Appl. Math. {\bf 397}(2021), 113611. 

\bibitem{Ross1916}
R. Ross, {\it {An application of the theory of probabilities to the study of a priori pathomety}}, Proc. R. Soc. A {\bf 92}(1916), 204-230 (Part I).

\bibitem{Ross1917}
R. Ross and H. Hudson, {\it {An application of the theory of probabilities to the study of a priori pathomety}} Proc. R. Soc. A {\bf 93}(1917), 212-225 (Part II), 225-240 (Part III).

\bibitem{slimopen} L. Van Schaik, D. Duives, S. Hoogendoorn-Lanser, J. Hoekstra, W. Daamen,  A. Gavriilidou, P. Krishnakumari, M. Rinaldi and S.P. Hoogendoorn. {\it {Understanding physical distancing compliance behaviour using proximity and survey data: A case study in the Netherlands during the COVID-19 pandemic}} Transportation Research Procedia (in press).


\bibitem{Schmidt1990}
B. Schmidt, {\it{Die epidemische Dynamik von Infektionskrankheiten mit multiplem Krankheitsverlauf in strukturierten Bev{\"o}lkerungen Mathematische Modellierung, Sensitivit{\"a}tsanalyse und numerische Simulation der Ausbreitung von AIDS}}, PhD thesis, Universit{\"a}t zu K{\"o}ln, 1990.

\bibitem{Tang2022}
B. Tang, W. Zhou, X. Wang, H. Wu and Y. Xiao, {\textit {Controlling multiple Covid-19 waves: an insight from a multi-scale model linking the behaviour change dynamics to disease transmission dynamics}}, Bull. Math. Biol. {\bf 84}(2022), 106.

\bibitem{Teslya2022}
A. Teslya, H. Nunner, V. Buskens and M. Kretzschmar, {\textit {The effect of competition between health opinions on epidemic dynamics}}, Proc. Natl. Acad. Sci. Nexus {\bf 1}(2022), 1-14.

\bibitem{Thieme1979}
H. Thieme, {\it {On a class of Hammerstein integral equations}}, Manuscr. math. {\bf 29}(1979), 49-84.

\bibitem{Thieme1980}
H. Thieme, {\it {On the boundedness and the asymptotic behaviour of the non-negative solutions to Volterra-Hammerstein integral equations.}}, Manuscr. math. {\bf 31}(1980), 379-412.

\bibitem{Thieme1984}
H. Thieme, {\it {Renewal theorems for linear periodic Volterra integral equations}}, J. Integral Equations {\bf 7}(1984), 253-277.

\bibitem{Thieme1985}
H. Thieme, {\it {Renewal theorems for some mathematical models in epidemiology}}, J. Integral Equations {\bf 8}(1985), 185-216.

\bibitem{Thieme2009a}
H. Thieme, {\it {Distributed susceptibility: a challenge to persistence theory in infectious disease models}}, Disc. Cont. Dyn. Sys. B {\bf 12}(2009), 865-864.

\bibitem{Thieme2009}
H. Thieme,  {\it {Spectral bound and reproduction number for infinite-dimensional population structure and time heterogeneity}}, SIAM J. Appl. Math. {\bf 70}(2009), 188-211.

\bibitem{Thieme2000}
H. Thieme and J. Yang, {\it {On the complex formation approach in modeling predator prey relations, mating, and sexual disease transmission}}, Elect. J. Diff. Eqns. {\bf 5}(2000), 255-283.

\bibitem{Tian2023}
Y. Tian, A. Sridhar, C. Wu, S. Levin, K. Carley, H.  Poor and O. Ya{\u{g}}an, {\it {Role of masks in mitigating viral spread on networks}}, Phys. Rev. E {\bf 108}(2023).

\bibitem{Tkachenko2021}
A. Tkachenko, S. Maslov, A. Elbanna, G. Wong, Z. Weiner and N. Goldenfeld, {\it{
 Time-dependent heterogeneity leads to transient suppression of the COVID-19 epidemic, not herd immunity}}, 
Proc. Natl. Acad. Sci. {\bf 118}(2021), e2015972118.

 \bibitem{Tkachenko2021b}A. Tkachenko, S. Maslov, T. Wang, A. Elbanna, G. Wong and N. Goldenfeld, {\it{ Stochastic social behavior coupled to COVID-19 dynamics leads to waves, plateaus, and an endemic state}}, eLife (2021) https://doi.org/ 10.7554/eLife.68341].

\bibitem{Toorians2021}
M. Toorians, A. MacPherson and T. Davies, {\textit{ Revisiting pathogen transmission in epidemiological models and its role in the disease-diversity relation}}, Preprints (2021), 2021100295. https://doi.org/10.20944/preprints202110.0295.v4. 

\bibitem{Wambua2023}
J. Wambua, N. Loedy, C. Jarvis, K. Wong, C. Faes, R. Grah, B. Prasse, F. Sandmann, R. Niehus, H. Johnson, W. Edmunds, P. Beutels, N. Hens and P. Coletti, {\it {The influence of COVID-19 risk perception and vaccination status on the number of social contacts across Europe: insights from the CoMix study}}, BMC Public Health {\bf 23}(2023).

\bibitem{Wong2020}
G. Wong, Z. Weiner, A. Tkachenko, A. Elbanna, S. Maslov and N. Goldenfeld, 
{\it {Modeling COVID-19 Dynamics in Illinois under Nonpharmaceutical Interventions}}, 
Phys. Rev. X {\bf 10}(2020), 041033.

\bibitem{Zhang2023}
X. Zhang, F. Scarabel, K. Murty and J. Wu, {\it {Renewal equations for delayed population behaviour adaptation coupled with disease transmission dynamics: A mechanism for multiple waves of emerging infections}}, Math. Biosci. {\bf 365}(2023).

\end{thebibliography}
\end{document}